\documentclass[journal,twoside,web]{IEEEcolor}
\usepackage{generic}

\usepackage{etoolbox}
\makeatletter
\@ifundefined{color@begingroup}%
  {\let\color@begingroup\relax
   \let\color@endgroup\relax}{}%
\def\fix@ieeecolor@hbox#1{%
  \hbox{\color@begingroup#1\color@endgroup}}
\patchcmd\@makecaption{\hbox}{\fix@ieeecolor@hbox}{}{\FAILED}
\patchcmd\@makecaption{\hbox}{\fix@ieeecolor@hbox}{}{\FAILED} 

\usepackage{amsmath}
\usepackage{amssymb}
\usepackage{tikz}
\usetikzlibrary{calc, shapes, backgrounds}
\usetikzlibrary{positioning,arrows.meta,bending}
\usepackage{epsfig}
\usepackage{verbatim}
\usepackage{graphicx}
\usepackage{algorithm}
\usepackage{algpseudocode}
\usepackage{algorithmicx}
\newtheorem{definition}{Definition}
\newtheorem{theorem}{Theorem}
\newtheorem{lemma}{Lemma}

\newtheorem{assumption}{Assumption}
\newtheorem{sassumption}{Standing Assumption}
\newtheorem{proposition}{Proposition}

\newtheorem{remark}{Remark}
\usepackage[noadjust]{cite} 
\usepackage[nodisplayskipstretch]{setspace}
\newcommand{\Rl}{\ensuremath{\mathbb{R}}}   
\newcommand{\Rlp}{\ensuremath{\mathbb{R}_{>0}}}
\newcommand{\Rlo}{\ensuremath{\mathbb{R}_{\geq 0}}}
\newcommand{\Zo}{\ensuremath{\mathbb{Z}_{\geq 0}}}
\newcommand{\Zp}{\ensuremath{\mathbb{Z}_{> 0}}}

\newcommand{\Ex}{\mathbb{E}}

\title{Transmission power policies for energy-efficient wireless control of nonlinear systems}

\author{Vineeth S. Varma, Romain Postoyan, Daniel E. Quevedo and Irinel-Constantin Mor\u{a}rescu
\thanks{V.S. Varma, R. Postoyan and I-C. Mor\u{a}rescu are with Universit\'e de Lorraine, CNRS, CRAN, F-54000 Nancy, France, {\small \tt  vineeth.satheeskumar-varma@univ-lorraine.fr}.  D.E. Quevedo is with the School of Electrical Engineering and Robotics, Queensland University of Technology, Brisbane, Australia, {\small  \tt  dquevedo@ieee.org}.\newline This work was supported by ANR through the grants HANDY no. ANR-18-CE40-0010 and NICETWEET, no. ANR-20-CE48-0009.}}

\begin{document}
\maketitle
\thispagestyle{empty}
\pagestyle{empty}
\begin{abstract} \textcolor{black}{We present an emulation-based controller and transmission policy design procedure for nonlinear wireless networked control systems.} The objective is to ensure the stability of the closed-loop system, in a stochastic sense, together with given control performance, while minimizing the average power used for communications. The controller is designed by emulation, i.e., ignoring the network, and the transmission power is given by threshold policies. These policies involve waiting a given amount of time since the last successful transmission instant, as well as requiring that the measured wireless channel gain is above a given threshold, before attempting a new transmission. Two power control laws are investigated: i) a constant power and ii) a power level inversely proportional to the channel gain. We explain how to select the waiting time, the channel threshold and the power level to minimize the induced average communication power, while ensuring the desired control objectives.
\end{abstract}


%
%
\section{Introduction}
This work aims at minimizing the energy consumption of wireless networks, which are being increasingly deployed in control systems \cite{ahlen2019toward}. Since 2011, about 2–6\% of the energy consumption worldwide arises from the communications and information industry, and a significant portion of this is contributed by the wireless and mobile communications companies \cite{wu2016green}. Improving the efficiency of this technology has therefore gained a rising amount of interest in recent years \cite{rault2014energy}. For mobile devices such as cellular phones, laptops, and mobile robots, smart and careful management of the energy utilized is essential due to the limited supply of energy available. For the case of fixed infrastructure connected to wireless networks, energy consumption has become a critical issue due to environmental and economic factors and has led to a large amount of research and publications \cite{hossain2012green,varma2013energy}.

In the wireless communication literature, various studies have investigated the design of energy-efficient communication systems to maximize the ratio of data rate to the energy consumed, or to minimize energy while maintaining a certain quality of service parameter, see \cite{rault2014energy} for an extensive survey. One of the most relevant techniques to improve energy efficiency is that of \emph{transmission power control}. In works like \cite{varma2013energy} and \cite{goodman2000power}, transmission power is optimized so that the ratio between the number of packets transmitted successfully to the power consumed is maximized. While these works are fully relevant in the context of regular communication systems, they are a priori not well-suited for wireless networked control systems (WNCS), which have different, specific requirements on control performance rather than maximizing data rates. 

\textcolor{black}{A few researchers have recently published results, which consider the problem above see e.g., \cite{de2010energy, rabi2008event,leong2017remote,quevedo2013power, li2013optimal}, with some of them utilizing power control \cite{leong2017remote,quevedo2013power, li2013optimal}. For example, an event-based power control policy using a threshold on the error covariance has been shown to perform optimally for state estimation in \cite{leong2017remote}. Energy-aware event-triggered strategies, in the sense that communications are only attempted when a state-dependent criterion holds, have recently been proposed for state-feedback controllers, see \cite{molin2009lqg,gatsis2014optimal,bala2020LQpower }. While event-based strategies are very promising, they require constant monitoring of the plant state (or output), which may be problematic in some set-ups for which time-triggered paradigms would be more appropriate.} Hence, when communication instants depend on time, instead of the state, results on communication energy minimization have been developed in \cite{varma2016energy}, assuming packets are always successfully transmitted but with varying costs, and in \cite{varma2019tac}, in which the average transmission power is minimized while ensuring the desired control performance for stochastic communication.
Even though recent works like \cite{maass2020stochastic} and \cite{maass2021stochastic} explore power control for interference management in nonlinear WNCS over static channels, results for nonlinear systems are crucially lacking and the design of transmission policies over a time-varying channel are missing even for linear systems.

In this work, we propose transmission power policies for nonlinear discrete-time systems controlled over a wireless network. For this purpose, we develop \emph{threshold-based} transmission policies, i.e., transmissions are not attempted until a certain threshold is passed on i) the time elapsed since the last successful communication and ii) the measured channel quality (or channel gain). Transmissions are attempted with a power level determined by the considered power policy until the packet is received as long as these conditions are satisfied. We consider both constant power policies and channel inversion policies, wherein the power level is inversely proportional to the channel gain. While inversion policies are in general more efficient, some communication devices and protocols may not allow the transmission power to be controlled freely. In such cases, constant power policies are a relevant alternative. The control law, on the other hand, is based on emulation, i.e., it is designed disregarding the presence of the wireless link to ensure the desired control objective. This allows the user to utilize their favorite discrete-time control methodology. In particular, we merely require the controller to be such that the origin of the closed-loop system is uniformly globally asymptotically stable, with a known Lyapunov function. Regarding the set-up, we investigate output-based control systems in which the wireless link is used to communicate information from either the sensor to the controller, or from the controller to the actuator, but not when both links are over a wireless network.

The main contributions are the following.
\begin{itemize}
    \item We formulate a framework for the design of threshold-based transmission policies for nonlinear discrete-time systems, in contrast to several works that focus on transmission policies for linear systems like \cite{gatsis2014optimal, varma2019tac}.
    \item We provide a set characterizing the usable length of the time interval before any transmission is attempted after a successful communication, the channel threshold, and the transmission power, which guarantee stability and a desired convergence rate of a given Lyapunov function in a stochastic sense. 
    \item We then observe that the minimization of the average communication power, while ensuring the desired control property, is a non-convex problem for both constant power and channel inversion policies. Consequently, we elucidate the following relevant sub-cases over which the minimization problem is solved: i) pure-time based in which the power control is independent of the channel quality, ii) pure-channel based in which the power control is independent of the time since the last successful communication,  iii) almost sure communications in which the channel threshold and transmission power are such that communication is almost always successful when attempted and finally iv) unsaturated polices in which the channel thresholds are such that channel inversion results in a transmission power smaller than the maximum allowable one.
\end{itemize} 
 

Compared to the preliminary version of this work presented in \cite{varma2020time}, which investigated purely time-based thresholds and constant power policies, in the present work, we additionally propose channel-based thresholds and channel inversion power policies, and account for a time-varying wireless channel, which is a more realistic assumption.

The rest of the paper is organized as follows: In Section~\ref{sec:sysm} we formally state the problem and the main assumptions considered. In Section~\ref{sec:stab}, we provide sufficient conditions to ensure the desired stochastic stability and performance properties of the WNCS. Next, in Section \ref{sec:comm}, we formalize the optimization problem under the constraint imposed by the stochastic stability and performance requirement, and then derive explicit solutions for relevant special cases. In Section~\ref{sect:assumption}, we elaborate on one of the standing assumptions stated in Section \ref{sec:sysm}. Finally, we provide numerical illustrations of our proposed communication strategy in Section \ref{sec:nr} before concluding in Section \ref{sec:concl}.
\noindent

\textbf{Notation.} Let $\Rl :=(-\infty,\infty)$, $\Rlo:=[0,\infty)$, $\mathbb{Z}_{>0}:=\{1,2,\ldots\}$ and $\Zo:=\{0,1,2,\ldots\}$. We use $\Pr(\cdot)$ for the probability and $\Ex[\cdot]$ for the expectation taken over the relevant stochastic variables. A function $\alpha : \Rlo \to \Rlo $ is of class $\mathcal{K}_{\infty}$ ($\alpha \in \mathcal{K}_\infty$) if it is continuous, strictly increasing, $\alpha(0)=0$ and $\lim_{s \to \infty} \alpha(s)= \infty$. For any $x_1\in \mathbb{R}^{n_1}$ and $x_2\in \mathbb{R}^{n_2}$ with $n_1,n_2 \in \Zp$, $(x_1,x_2)$ stands for $(x_1^\top,x_2^\top)^\top \in \mathbb{R}^{n_1+n_2}$.


\section{Problem statement}
In this section, we first describe the plant and controller model, followed by the communication model, the threshold policies and finally the objectives.

\label{sec:sysm}
\subsection{Plant and controller model}
We consider the discrete-time plant model given by
\begin{equation}
\begin{array}{rllll}
x_p(t+1) & = & f_p(x_p(t),u(t))\\
y(t) & = & g_p(x_p(t)),
\end{array}
\label{eq:mainsys}
\end{equation}
where $t \in\Zo$ is the time, $x_p(t) \in \mathbb{R}^{s_p}$ is the plant state, $u(t)\in \mathbb{R}^{s_u}$ is the control input, $y(t) \in \mathbb{R}^{s_y}$ is the measured output used for control and $s_p,s_u,s_y \in \mathbb{Z}_{>0}$ are their respective dimensions.

We proceed by emulation, and thus assume that we know a stabilizing output-feedback controller for system \eqref{eq:mainsys} of the form
\begin{equation}
\begin{array}{rllll}
x_c(t+1) & = & f_c(x_c(t),y(t))\\ 
u(t)& = & g_c(x_c(t),y(t)),
\end{array}
\label{eq:maincontrol}
\end{equation}
where $x_c(t) \in \mathbb{R}^{s_c}$ is the controller state. When the controller is static, we simply have $u(t)=g_c(y(t))$ in \eqref{eq:maincontrol}.  At this stage, any controller design techniques can be employed to construct (\ref{eq:maincontrol}), like backstepping, feedback linearization etc.  The assumption we make on the closed-loop system (\ref{eq:mainsys})-(\ref{eq:maincontrol}) is formalized in the sequel. 

\tikzstyle{block} = [draw, fill=blue!20, rectangle, 
    minimum height=6em, minimum width=5em]
\tikzstyle{sum} = [draw, fill=blue!20, circle, node distance=1cm]
\tikzstyle{input} = [coordinate]
\tikzstyle{output} = [coordinate]
\tikzstyle{pinstyle} = [pin edge={to-,thin,black}]

    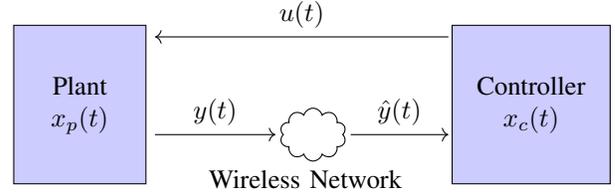
\begin{figure}[t]
    \centering
				\centering
		\begin{tikzpicture}[auto,node distance=2cm]
   
    \node at (0,0) [block] (controller) {\begin{tabular}{c} Plant \\ $x_p(t)$ \end{tabular}};
    \node at (6,0)[block] (system) {\begin{tabular}{c}
        Controller\\ $x_c(t)$   
    \end{tabular}};
    \node at (3.1,-.4) [cloud,draw, minimum height=2em, minimum width=2.5em] { };
    \node at (3,-1) [rectangle] {Wireless Network};
    \draw [->] (4.9,0.9)  --  node[above]{$u(t)$} (1,0.9);
      \draw [->] (1,-.4)  --  node{$y(t)$} (2.6,-0.4);
      
        \draw [->] (3.6,-.4)  --  node{$\hat{y}(t)$} (4.9,-0.4);

\end{tikzpicture}
    \caption{Schematic of the networked control system.}
    \label{fig:ncs}
    \end{figure}
    
We are interested in the scenario where plant \eqref{eq:mainsys} and controller \eqref{eq:maincontrol} communicate over a wireless channel as illustrated in Figure \ref{fig:ncs}, specifically, the wireless link is used to communicate information from the sensors to the controller. As a result, the feedback loop is no longer closed at every time instant $t \in \Zo$, but only at the instants $t_k \in \mathcal{T} \subseteq \Zo, k \in \mathbb{Z}_{>0}$ when communication is successful. In the absence of communication, the controller uses a so-called \emph{networked version} \cite{nesic2004input} of the output measurement denoted by $\hat{y}$. Controller \eqref{eq:maincontrol} becomes in this context
\begin{equation}
\left(\begin{array}{c}
  x_c(t+1) \\
    u(t) \end{array}\right)
    =   \left\{ \hspace{-0.2cm}
\begin{array}{l}
 \left(\begin{array}{l} f_c(x_c(t),y(t)) \\ 
 g_c(x_c(t),y(t)) 
 \end{array}\right)
   \text{ for } t \in \mathcal{T}\\
 \left(\begin{array}{l} f_c(x_c(t),\hat{y}(t)) \\
 g_c(x_c(t),\hat{y}(t)) 
 \end{array}\right)  \text{ for } t \in \Zo \setminus \mathcal{T}.
\end{array}\right. 
\label{eq:maincontrol2}
\end{equation}
The networked version of the output $\hat{y}$ generated at the controller evolves according to the following dynamics
\begin{equation}
\hat{y}(t+1)  = \left\{
\begin{array}{ll}
 \hat{f}(g_p(x_p(t))) & \text{ if } t \in \mathcal{T}\\
 \hat{f}(\hat{y}(t)) & \text{ if } t \in \Zo \setminus \mathcal{T},
\end{array}\right.
\label{eq:holding}
\end{equation}
where $\hat{f}$ is the holding function applied, which can take various forms including the zero-order-hold strategy $\hat{f}(\hat{y})=\hat{y}$, or the zeroing policy $\hat{f}(\hat{y})=0$ for any $\hat{y} \in \mathbb{R}^{s_y}$. Note that $\hat y$ is never reset to the actual value of $y$ in \eqref{eq:holding}. This is in accordance with the way we model the closed-loop system in the sequel, in which $u(t)$ and $x_c(t+1)$ depends on $(x_c(t),y(t))$ when the packet is successfully received and on $(x_c(t),\hat{y}(t))$ only when the network packet is lost.

\begin{remark}
The results presented in this paper apply \emph{mutatis mutandis} when the network is located between the controller and the actuator, and not between the sensors and the controller as in Figure \ref{fig:ncs}, by changing the network variable to be $\hat{u}$ instead of $\hat{y}$. When the network is used in both directions, the analysis becomes quite convoluted, especially if communication events occur independently; this case is left for the future.  \hfill $\Box$
\end{remark}

Based on \eqref{eq:mainsys}-\eqref{eq:holding}, we introduce the concatenated state  $\chi:= (x_p,x_c,\hat{y}) \in \mathbb{R}^{s_{\chi}}$ with $s_{\chi}:=s_p+s_c+s_y$, and we write the closed-loop dynamics of the WNCS as
\begin{equation}
\chi(t+1) = \left\{
\begin{array}{ll}
 f_S(\chi(t)) & \text{ for } t \in \mathcal{T}\\
 f_U(\chi(t)) & \text{ for } t \in \Zo \setminus \mathcal{T},
\end{array}\right.
\label{eq:mainsys2}
\end{equation} 
where $f_S,f_U$ are defined as
\begin{equation}
f_S(\chi) :=  \left( \begin{array}{c}
      f_p(x_p,g_c(x_c,g_p(x_p)))  \\
      f_c(x_c,g_p(x_p))  \\
      \hat{f}(g_p(x_p))
  \end{array}  \right) , \label{eq:deffs}
\end{equation}
and 
\begin{equation}
    f_U(\chi) :=  \left( \begin{array}{c}
      f_p(x_p,g_c(x_c,\hat{y}))  \\
      f_c(x_c,\hat{y})  \\
      \hat{f}(\hat{y})
  \end{array}  \right). \label{eq:deffu}
\end{equation}
The standing assumptions (SA) we make on system \eqref{eq:mainsys2} are stated next.

\begin{sassumption}[SA1]
There exist $\overline{\alpha}, \underline{\alpha} \in \mathcal{K}_{\infty}$, $a_S \in [0,1)$, $a_U> a_S$ and $V: \mathbb{R}^{s_{\chi}} \to \Rlo$ such that, for any $\chi\in\mathbb{R}^{s_{\chi}}$,
\begin{subequations}
\begin{align}
\underline{\alpha}(|\chi|) \leq V(\chi) \leq \overline{\alpha}(|\chi|)\label{eq:ass-sandwich-bounds}\\
V(f_S(\chi)) \leq a_S V(\chi), \label{eq:a1ass1}\\
V(f_U(\chi)) \leq a_U V(\chi). \label{eq:a0ass1}
\end{align}
\end{subequations}
\hfill $\Box$
\label{ass:alpha1}
\end{sassumption}

Properties \eqref{eq:ass-sandwich-bounds} and \eqref{eq:a1ass1} imply that the origin of system $\chi(t+1)=f_S(\chi(t))$ is uniformly globally asymptotically stable (UGAS). This is typically the case when controller \eqref{eq:maincontrol} has been designed to ensure that the origin of system \eqref{eq:mainsys}-\eqref{eq:maincontrol} is UGAS, see Section \ref{sect:assumption}. The fact that the bound in \eqref{eq:a1ass1} is linear in $V$ comes with no loss of generality. Indeed, if we know a Lyapunov function which does not admit a linear bound as in \eqref{eq:a1ass1}, we can always modify it to satisfy \eqref{eq:ass-sandwich-bounds} and \eqref{eq:a1ass1}, under mild regularity assumptions, see Theorem 2 in \cite{hespanha2008lyapunov}. On the other hand, \eqref{eq:a0ass1} in SA\ref{ass:alpha1} imposes a condition on the growth rate of $V$ along solutions to \eqref{eq:mainsys2} when a transmission fails. Typically $a_U$ is strictly larger than $1$, and we assume $a_S < a_U$ implying that successful communications improve the guaranteed convergence speed of the Lyapunov function $V$ to zero, along the solutions to \eqref{eq:mainsys2}. Conditions ensuring the satisfaction of SA\ref{ass:alpha1} are discussed in more details in Section \ref{sect:assumption}, where we show that SA\ref{ass:alpha1} can always be ensured for detectable and stabilizable linear time-invariant systems.  

To conclude the description of the closed-loop system \eqref{eq:mainsys2}, we need to explain when a communication attempt is successful or not. 

\subsection{Communication setup}

In this sub-section, we describe the sequence of successful communication instants $t_k \in \mathcal{T}$. In wireless communication, the \emph{signal-to-interference plus noise ratio} (SINR) determines the probability of successful communication. The SINR is determined by 
\begin{enumerate}
    \item[i)] the transmission power $P(t) \in [0,P_{\max}]$ at time $t \in \Zo$, with $P_{\max}>0$ being the maximum transmission power allowed by the transmitter at any time,
    \item[ii)] the channel gain, which is an exogenous time-varying parameter,
    \item[iii)] and the power of the white noise, which is a constant we normalize to $1$.
\end{enumerate}

The channel gain is typically estimated by a feedback from the receiver after the transmitter sends pilot signals, which costs the transmitter some power. The estimated value of this quantity, which we call the \emph{channel measurement} (CM), is denoted by $h(t)$. In some cases, like in carrier-sense multiple access (CSMA), where the channel gain is used to represent the amount of interference in the medium, the transmitter simply senses the wireless medium to check for interference and this will not cost the transmitter any power. We use $q(t) \in \{0,1\}$ to express if the channel was estimated at time $t \in \Zp$ (indicated by $q(t)=1$) or not (indicated by $q(t)=0$). We make the following assumption for the CM, which is relaxed later in Section \ref{sec:TTnocsi}.

\begin{sassumption}[SA2]
 For any $t \in \Zo$, the CM $h(t) \in \mathcal{H}$, with $\mathcal{H}$ being a finite set, and it is exactly obtained by the transmitter when $q(t)=1$ by spending a fixed amount of power $P_S \in \Rlo$. \label{ass:estimateh} \hfill $\Box$
\end{sassumption}

Next, we make the following assumption regarding the probability of successfully receiving the packet at time $t \in \Zo$. 

\begin{sassumption}[SA3] The following holds.
\begin{enumerate}
    \item[(i)] The packet success rate, i.e., the probability of the communication attempt succeeding, is given by a known function $\psi(P(t)h(t))$, where $\psi: \Rlo \to [0,1]$. The mapping $\psi$ is: (i-a) differentiable, (i-b) strictly increasing on $\Rlo$, (i-c) initially convex and then concave, (i-d) $\psi(0)=0$ and $\lim_{\gamma \to \infty} \psi(\gamma)=1$.
    \item[(ii)] When a packet sent at time $t \in \mathbb{Z}_{>0}$ is received, the transmitter obtains an acknowledgement  before $t+1$ without any error. 
    \item[(iii)] The CM $h(t)$ is an i.i.d. random variable with a known probability distribution $\rho$, i.e., $\rho(h)=\Pr(h(t)=h)$ for all $h \in \mathcal{H}$. \hfill $\Box$
\end{enumerate} \label{ass:CSITpd}
\end{sassumption}

Item (i) of SA\ref{ass:CSITpd} models the packet error rate as a smooth time-invariant function of the transmission power, as is common in wireless communication literature \cite{wu2016green, varma2013energy}. The additional properties considered are quite standard in wireless literature, see \cite{goodman2000power}, \cite{rodriguez2003analytical} for example. On the other hand, most practical communication setups like, e.g., Wifi, 4G and 5G use some sort of ACK protocol so that item (ii) of SA\ref{ass:CSITpd} is reasonable. The ACK packets have a size of the order of a few bits and are typically much smaller than the control/output information packets, and can thus be assumed to be received without any loss \cite{varma2013energy}. \textcolor{black}{On the other hand, a simple (but conservative) way to incorporate ACK packet losses into our framework would be to include the ACK packet loss in the expression of $\psi$. This means that the communication will be seen as a failure if the ACK packet is not received. We will also see in Section \ref{sec:purech}, a transmission policy that does not require ACK signals to be implemented, thereby relaxing item (ii) of SA\ref{ass:CSITpd}. Finally, the channel gain is often assumed to be i.i.d. in wireless engineering, see Chapter 5 of the book on wireless communications in practice \cite{rappaport1996wireless}. Item (iii) of SA\ref{ass:CSITpd} follows as the CM is simply a quantization of the channel gain.}
\subsection{Threshold  policies}

We focus on \emph{threshold-based transmission} policies that determines the transmission power $P(t)$ at each instant $t \in \Zp$. In particular, we impose a threshold on the time steps since the last successful transmission and on the CM $h(t)$. The former implies that communication is attempted only when a certain number of time instants have elapsed since the last successful communication, which is known by the transmitter in view of item (ii) in SA\ref{ass:CSITpd}. To model this number, we introduce the clock $\tau(t)\in \Zp$ for all $t \in \Zp$, which counts the number of time instants elapsed since the last successful communication as follows
\begin{equation}\label{eq:taud}
\tau(t+1) =  \left\{ \begin{array}{ll}
1 & \text{for } t \in \mathcal{T} \\
\tau(t)+1 \hfill & \text{for } t \in \Zo \setminus \mathcal{T}.\\
\end{array} \right.
\end{equation}
We assume that the initial time is a successful communication instant, i.e., we set $t_1=0$ resulting in $0 \in \mathcal{T}$ and $\tau(0)=1$.

We use $\mathcal{P}: \Rlo \to [0,P_{\max}]$ to denote the power control function, which will be designed in Section \ref{sec:stab}, i.e., the transmission power used when the CM $h(t)$ is above a certain threshold $\bar{h}$. The considered class of transmission policies can then be written for any $t \in \Zp$ as
\begin{equation}
    P(t) = \left\{ \begin{array}{ll}
      \mathcal{P}(h(t)) &  \text{if }  h(t) \geq \bar{h} \text{ and } \tau(t) \geq n+1\\
       0  & \text{otherwise.} 
    \end{array} \right.
    \label{eq:TCTPC} 
\end{equation}
Policy \eqref{eq:TCTPC} does precisely what we stated, i.e., communication is triggered only when both the time since the last transmission $\tau(t)$ and the CM $h(t)$ are above given thresholds $n$ and $\bar{h}$ respectively, which constitute design parameters. Since communication is never attempted when $\tau(t) \leq n$, we do not need to spend $P_S$ to estimate the channel for these time instants, see SA\ref{ass:estimateh}.
Under policy \eqref{eq:TCTPC}, the WCNS \eqref{eq:mainsys2} becomes

\begin{equation}
\left( \begin{array}{c}\chi(t+1) \\ \tau(t+1) \end{array}\right)= \left\{
\begin{array}{ll}
\left(\begin{array}{c}  f_S(\chi(t)) \\ 1 \end{array} \right) &\hspace{-.3cm}  \begin{array}{l}  \text{if } \tau(t) \geq n+1 \text{ and }\\
h(t) \geq \bar{h}, \text{ with}  \\\text{probability }\\ \psi (h(t)\mathcal{P}(h(t))) , \end{array} \\
 \left(\begin{array}{c}  f_U(\chi(t)) \\ \tau(t)+1 \end{array} \right)  &\hspace{-.3cm}  \text{ otherwise.}
\end{array}\right.
\label{eq:mainsys3b}
\end{equation} 
Recall that the probability of successful communication, when the thresholds are satisfied at time and transmissions are attempted, is given by $\psi(h(t)P(t))$ from item (i) in SA\ref{ass:CSITpd}. 

\begin{remark}
In a more general setting, one could design $P(\cdot)$ as a function of $\tau(\cdot)$ as is done in \cite{varma2019tac} for linear systems and also as a function of $h(\cdot)$. However, the objective of this work is to focus on threshold policies as described in \eqref{eq:TCTPC}, which are easier to design and implement, and have proved their strengths/relevance in the context of estimation and wireless communication \cite{goodman2000power, varma2013energy}. \hfill $\Box$
\end{remark}

\color{black}
\subsection{Objectives}

The first objective of this work is to preserve the stability of the WNCS. Due to the stochastic nature of communication success, we can no longer ensure the original UGAS property guaranteed by SA\ref{ass:alpha1}. Instead, we rely on the stochastic notion of stability defined next, which is inspired from \cite{quevedo2014stochastic}.

\begin{definition}
We say that the set $\{(\chi, \tau): \chi=0 \}$ is \emph{stochastically stable} for system \eqref{eq:mainsys3b}, if there exists $\alpha \in \mathcal{K}_{\infty}$, such that for any solution $(\chi,\tau)$ to \eqref{eq:mainsys3b}, 
\begin{equation}
   \sum_{t=0}^{\infty} \Ex[\alpha(|\chi(t)|)] < \infty. \label{eq:ststab}
\end{equation}
\hfill $\Box$
\label{def:ss} 
\end{definition}

Definition \ref{def:ss} implies that we are merely interested in the stability of the origin for $\chi$, and not $\tau$, which is simply constructed to count the time since the last transmission. In addition to the partial stability property described above, we also want to ensure that the Lyapunov function $V$ in SA\ref{ass:alpha1} converges in expectation, with a certain given rate $\mu \in (a_S,\min\{1,a_U\})$, along solutions to \eqref{eq:mainsys3b}, i.e.,
\begin{equation}
  \Ex[V(\chi(t))] \leq \mu^t V(\chi(0)) \label{eq:stmu}
\end{equation}
for any solution $(\chi,\tau)$ to \eqref{eq:mainsys3b} for all $t \in \Zp$. Property \eqref{eq:stmu} serves as a measure of the control performance of system \eqref{eq:mainsys3b} and satisfying it automatically ensures \eqref{eq:ststab} as $\mu<1$ in view of \eqref{eq:ass-sandwich-bounds}. Note that we always pick $\mu<a_U$ as otherwise, never communicating would achieve the objective in \eqref{eq:stmu}. 

An intuitive way to ensure the two above properties is to set $P(t)=P_{\max}$ for all $t\geq 0$ by taking $n=0$ and $p=P_{\max}$. This would result in frequent successful communications in view of item (i) of SA\ref{ass:CSITpd}, but also, and importantly, in a high power consumption \cite{wu2016green}. To overcome this potential issue, we want to reduce the average power consumed while satisfying the convergence property \eqref{eq:stmu} (and thereby ensuring \eqref{eq:ststab}). The average communication power over an infinite horizon is defined as
\begin{equation}
 J(\mathbf{P}) := \lim_{T \to \infty}\frac{1}{T} \mathbb{E}\left[ \sum_{t=1}^{T} P(t) + P_s q(t)\right], \label{eq:costperk}
\end{equation}
where $\mathbf{P}=(P(1),P(2),\dots)$ is the sequence of transmission powers applied at instances dictated by the threshold policy. Our objective is to find the optimal $\bar{h}$ and $n$ for the two types of power control policies detailed in Section \ref{sec:stab}, taking into account \eqref{eq:stmu} and \eqref{eq:costperk}. \textcolor{black}{Note that reducing communications may result in a deterioration in control performance. Our approach in handling this trade-off is to reduce the communication cost as much as possible, while ensuring a certain level of control performance determined by $\mu$. This parameter $\mu$ is tunable and can be selected to fit the demands of the intended application as is illustrated later in Section \ref{sec:tradeoff}. }


\section{Stochastic stability and control performance}
\label{sec:stab}

In this section, we first provide conditions on $n$, the time threshold used in \eqref{eq:TCTPC}, and the probability of successful communication to ensure the stability property \eqref{eq:stmu}. Afterwards, we clarify how this probability of successful communication depends on $\mathcal{P}$ and $\bar{h}$ for constant power and channel inversion policies.

For our analysis, it is important to note that $h(t)$ is assumed to be i.i.d. in view of item (iii) in SA\ref{ass:CSITpd}, and the power control function $\mathcal{P}(h(t))$ only depends on this variable. Therefore, given a channel threshold $\bar{h}$ and power control function $\mathcal{P}(h(t))$, the probability of successful communication when $\tau(t) \geq n+1$ is fixed over all channel realizations.  We use $\eta$ to denote this probability, where, given $\bar{h}$ and $\mathcal{P}(\cdot)$, 
\begin{equation}
    \eta: = \sum_{h \in \mathcal{H}, h \geq \bar{h}} \psi(\mathcal{P}(h) ) \rho(h) .
\end{equation}  

\subsection{Stability conditions for a given $\eta$}

Given a convergence rate $\mu \in (a_S,\min\{1,a_U\})$ for the expected value of $V$ as in \eqref{eq:stmu}, and $n \in \Zo$, we first identify a set of probabilities $\eta$'s ensuring \eqref{eq:stmu} and stochastic stability as in Definition~\ref{def:ss}. 

For that purpose, we define for any $\eta \in [0,1]$ and $n \in \Zo$, the convergence rate function
\begin{equation}
   \beta(n,\eta) :=\exp\left( \frac{\eta\log(a_S a_U^n) + \log(a_U) (1-\eta) }{ 1+ n\eta}\right), \label{eq:muofp}
\end{equation}
where $a_S$ and $a_U$ come from SA\ref{ass:alpha1}. We next provide conditions on $n$ and $\eta$ to ensure the desired stability property \eqref{eq:stmu}.

\begin{proposition}
Consider $\mu \in (a_S,\min\{1,a_U\})$, $n \in \Zo$ and $\eta \in [0,1]$, if $\beta(n,\eta) \leq \mu$, then the WNCS \eqref{eq:mainsys3b} is stochastically stable and
\begin{equation}
    \Ex[V(\chi(t))] \leq \beta(n,\eta)^t V(\chi(0)) \leq \mu^t V(\chi(0))
\end{equation}
for any solution $(\chi,\tau)$ to \eqref{eq:mainsys3b} at any $t \in \Zo$ . \hfill \label{prop:etastab} $\Box$
\end{proposition}
\begin{proof}
See Appendix \ref{app:p1}
\end{proof}
Proposition \ref{prop:etastab} implies that, as long as the chosen $n$ and the resulting $\eta$ are such that $\beta(n,\eta) \leq \mu$,  with $\beta(n,\eta)$ defined in \eqref{eq:muofp}, the desired stability and convergence properties are ensured. 

Proposition \ref{prop:etastab} cannot be directly exploited to design transmission power policies as it involves $\eta$, which depends on the channel threshold $\bar{h}$ as well as on the power control function $\mathcal{P}$. \textcolor{black}{When additional properties on $F$ and $V$ in SA\ref{ass:alpha1} are known, less conservative bounds on $\beta$ which characterizes the growth of $V$ may be derived.}.  We explain how to obtain $\eta$ for the considered transmission policies in the following.

\subsection{Evaluation of $\eta$ for the considered power policies}

As already mentioned, we focus on two types of power control policies:  i) constant power and ii) channel inversion policies. 

\subsubsection{Constant power policy} This policy implies that, whenever the conditions for transmission are satisfied according to \eqref{eq:TCTPC}, communication is attempted with a constant power $p \in [0,P_{\max}]$, i.e., $ \mathcal{P}(h)=p$ for any $h \geq \bar{h}$ with $ h \in \mathcal{H}$. The value of this constant power is a design parameter. 

\begin{lemma}
Given $\bar{h} \in \mathcal{H}$ and $p \in [0,P_{\max}]$, the probability of successful transmission at any time $t$ for which $\tau(t) \geq n+1$ under the constant power policy using power $p \in [0, P_{\max}]$ is given by
\begin{equation}
   \eta_C(\bar{h},p):= \sum_{h \in \mathcal{H}, h \geq \bar{h}} \psi(ph) \rho(h) .\label{eq:etaCP} 
\end{equation} \label{lem:etaCP} \hfill $\Box$
\end{lemma}
\begin{proof}
Recall that item (i) of SA\ref{ass:CSITpd} gives the probability of success for a given CM $h(t)$ to be $\psi(p h(t))$ while using power $p$. Due to the channel threshold we impose, transmissions occur only when $h(t) \geq \bar{h}$. Since $\rho(h)= \Pr(h(t)=h)$, which is an i.i.d. random variable, the expected packet success rate is given by \eqref{eq:etaCP}.
\end{proof} 

\subsubsection{Channel inversion policy} The second policy we consider is described by
\begin{equation}
    \mathcal{P}(h):= \min\left\{ P_{\max}, \displaystyle{ \frac{\kappa}{h} } \right\} ,\label{eq:chinversion}
\end{equation}
where $\kappa>0$ is the power gain. When the threshold conditions in \eqref{eq:TCTPC} are met, this policy applies a transmission power which is inversely proportional to the channel gain $h(t)$ if feasible, i.e., when $ \frac{\kappa}{h} \leq  P_{\max}$, and $P_{\max}$ otherwise. Channel inversion is often used in wireless communication to maintain a certain SINR at the receiver, and sometimes to optimize the transmission power \cite{goodman2000power, shen2003optimal}. Figure 2 illustrates the power used for a given CM after applying \eqref{eq:chinversion} and \eqref{eq:TCTPC} with $\kappa=2$, $\bar{h}=1$ and $P_{\max}=1$. 
\begin{figure}
    \centering
    \hspace{-3mm}\includegraphics[width=0.53\textwidth,trim={0 9cm 0 9cm},clip]{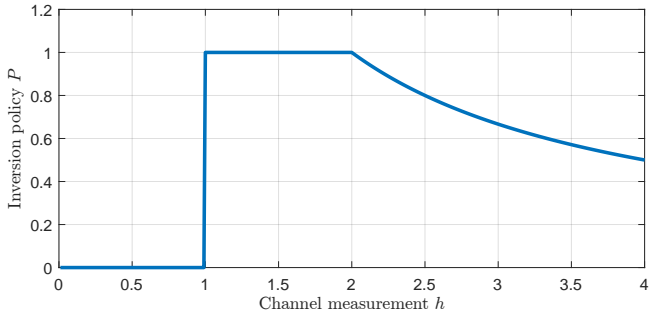}
    \caption{Transmission power $P$ for a given CM $h$ using \eqref{eq:TCTPC} and \eqref{eq:chinversion} with $\bar{h}=1,P_{\max}=1$ and $\kappa=2$.}
    \label{fig:my_label}
\end{figure}

First, we provide the expression for the probability of successful communication when $\tau(t)\geq n+1$, denoted by $\eta_I$, in the next lemma.
\begin{lemma}
Given $ \bar{h} \in \mathcal{H}$ and $ \kappa \in \Rlo$, the probability of successful transmission at any time $t$ for which $\tau(t) \geq n+1$, under the channel inversion policy \eqref{eq:chinversion} is given by
      \begin{equation} \begin{array}{ll}
        \eta_I(\bar{h},\kappa) =   & \left\{   \begin{array}{l}
             \Pr(h(t) \geq \bar{h}) \psi(\kappa)
             \text{ when } \bar{h} \geq \frac{\kappa}{P_{\max}},\\
             \Pr(h(t) > \frac{\kappa}{P_{\max}}) \psi(\kappa)
           \\  + \sum_{h \in \mathcal{H}, \bar{h} \leq h \leq \frac{\kappa}{ P_{\max}} } \psi(h P_{\max})\rho(h)\\
           \text{otherwise.}
        \end{array} \right. \label{eq:psr2} 
      \end{array}
\end{equation} \label{lem:IP} \hfill $\Box$
\end{lemma}
\begin{proof}
Since, $h(t)$ is i.i.d. and $\Pr(h(t)=h)=\rho(h)$ for any given $t \in \Zo$, we have identical probabilities for the channel distribution and also for successful transmissions when $\tau(t) \geq n+1$.

When $\bar{h} \geq \frac{\kappa}{P_{\max}}$, $\mathcal{P}(h)$ is always given by $\frac{\kappa}{h}$ as $\frac{\kappa}{h}$ will always be smaller than $P_{\max}$. When $h(t)\geq \bar{h}$, we obtain that $\mathcal{P}= \frac{\kappa}{h(t)}$ and therefore, the probability of packet success is given by $\psi(\kappa)$. Since transmissions are only attempted when $h(t) \geq \bar{h}$, we have the probability of successful transmissions is given by the first case of \eqref{eq:psr2}.

On the other hand, if $\bar{h} > \frac{\kappa}{P_{\max}}$, we have two possible cases. The first is when $ \mathcal{P}(h) < P_{\max}$, which following the previous logic allows us evaluate the probability of successful transmission as $\psi(\kappa)$. This happens with a probability $\Pr(h(t) > \frac{\kappa}{P_{\max}})$ and gives us the first line of \eqref{eq:psr2}. Secondly, if $h(t)>\bar{h}$ but $\frac{\kappa}{h(t)} \geq P_{\max}$, the probability of successful transmission is simply given by $\psi(h(t) P_{\max})$ due to the transmission power saturation. We can evaluate the expectation of this probability over the relevant limits of $h(t)$ to derive the second line of \eqref{eq:psr2}. 
\end{proof}
\subsection{Main result}
We are ready to state the main stability result. It provides a condition on the parameters $n, \bar{h}$ and $p$ or $\kappa$ for \eqref{eq:stmu} to hold under the constant power or channel inversion policies respectively.

\begin{theorem}
For given $\mu \in (a_S,\min\{1,a_U\})$ and $n \in \Zo$, the WNCS \eqref{eq:mainsys3b} is stochastically stable and \eqref{eq:stmu} holds for
\begin{itemize}
    \item constant power policies if $\beta(n,\eta_C(\bar{h},p)) \leq \mu$, with $\eta_C$ from \eqref{eq:etaCP},
    \item channel inversion policies if $\beta(n,\eta_I(\bar{h},\kappa)) \leq \mu$, with $\eta_I$ from \eqref{eq:psr2}.\hfill $\Box$ \label{th:sstab}
\end{itemize} 
\end{theorem}
\begin{proof}
The proof directly follows from Proposition \ref{prop:etastab} and using Lemmas \ref{lem:etaCP} and \ref{lem:IP} to get the corresponding values of $\eta$ for constant power and channel inversion policies respectively.
\end{proof}

Theorem \ref{th:sstab} provides conditions on $n, \bar{h},p$ and $\kappa$ to ensure the desired control properties. However, these conditions cannot be exploited directly in order to minimize the average communication cost \eqref{eq:costperk} in Section \ref{sec:comm} (while ensuring the desired control properties). This is because the set of $n, \bar{h},p$ and $\kappa$, such that its elements ensure the satisfaction of the conditions in Theorem \ref{th:sstab}, are not explicitly provided. Since $n$ and $\bar{h}$ belong to discrete sets, we will characterize the set of feasible $p$'s and $\kappa$'s, referred to as \emph{feasibility sets} in the subsection below.

\subsection{Feasibility sets}
\label{sec:fsets}
In order to identify the set of $p$'s and $\kappa$'s that ensure $\beta(n,\eta)\leq \mu$ for a given  $\bar{h}$, as required in Theorem~\ref{th:sstab}, we first provide the following lemma.

\begin{lemma}
For any given $\mu \in (0,1)$ and $n \in \Zo$, $\beta(n,\cdot)$ is decreasing. Additionally, $\eta_C(\bar{h},\cdot)$ and $\eta_I(\bar{h},\cdot)$ are increasing functions.\hfill $\Box$
\label{prop:minpar}
\end{lemma}
\begin{proof}
See Appendix \ref{app:p3}.
\end{proof} 

Given $\mu \in (a_S,1), n \in \Zo,\bar{h} \in \mathcal{H}$, we denote by $\underline{p}(\mu,n,\bar{h})$, the smallest solution to
 \begin{equation}
  \beta(n,\eta_C(\bar{h},\underline{p}(\mu,n,\bar{h})))=\mu. \label{eq:defup}
\end{equation}
In view of Lemma \ref{prop:minpar}, we have that $\beta(n,\eta_C(\bar{h},p)) \leq \mu$ for any $p \geq \underline{p}(\mu,n,\bar{h})$. If no such $\underline{p}(\mu,n,\bar{h})$ exists or if $ \underline{p}(\mu,n,\bar{h}) > P_{\max}$, then $  \beta(n,\eta_C(\bar{h},p))> \mu$ for all $p \in [0,P_{\max}]$ and the conditions of Theorem \ref{th:sstab} cannot hold. Otherwise, the set of feasible powers is identified as $[\underline{p}(\mu,n,\bar{h}), P_{\max} ]$

Similarly, denote by $\underline{\kappa}(\mu,n,\bar{h})$, the smallest solution to
 \begin{equation}
  \beta(n,\eta_I(\bar{h},\underline{\kappa}(\mu,n,\bar{h})))=\mu.\label{eq:defuk}
\end{equation}
In view of Lemma \ref{prop:minpar}, we have that $\beta(n,\eta_I(\bar{\kappa},p)) \leq \mu$ for any $\kappa \geq \underline{\kappa}(\mu,n,\bar{h})$. If no such $\underline \kappa$ exists, then $  \beta(n,\eta_I(\bar{h},\kappa))> \mu$ for all $\kappa \in \Rlo$. Otherwise, the set of feasible $\kappa$'s is given by $[\underline{\kappa}(\mu,n,\bar{h}), \infty) $. Based on these observations, we are ready to proceed with the optimization of the transmission policies.

\section{Policy design}
\label{sec:comm}

In this section, we first provide the general expression of the communication cost under the constant power policy and the channel inversion policy. Minimizing this general cost in general is observed to be challenging due to the non-convex property of the cost function, and due to the average packet success $\eta$ being hard to evaluate. Indeed, while Lemmas \ref{lem:etaCP} and \ref{lem:IP} provide theoretical methods to evaluate the probability of successful communication when $\tau(t) \geq n+1$, in practice, the summation in \eqref{eq:etaCP} or \eqref{eq:psr2} is difficult to analyze. For example, consider that we use quadrature phase shift keying (QPSK) modulation, and an additive white Gaussian noise (AWGN) is present yielding $\psi(hp)=1-(1-0.5 \text{erfc}(\sqrt{hp}))^M$ where $M$ is the packet size. The summation of such an expression over $h \in \mathcal{H}$ becomes hard to characterize in terms of derivatives, convexity, etc. On the other hand, the first case of \eqref{eq:psr2} can be easily expressed analytically for several types of channel fading models that are often considered in wireless literature. For example, we have $Pr(h(t) \geq \bar{h}) = \exp(-\frac{\bar{h}}{2\sigma^2})$ for Rayleigh fading with $\sigma^2 \in \Rlp$ a constant parameter of the distribution.

All of these reasons motivate us to consider some special, relevant cases for the selection of time and channel thresholds, which allows for an easier evaluation of $\eta$ and minimization of the communication cost \eqref{eq:costperk}, while ensuring the stability and convergence properties as stated in Section \ref{sec:stab}. 

\subsection{General cost minimization}

First, we provide the expression of the cost function \eqref{eq:costperk} under the threshold policies we are interested in. 

\begin{proposition}
Under policy \eqref{eq:TCTPC}, the average communication cost in \eqref{eq:costperk} for constant power policies is given by
\begin{equation}
  J_{\mathrm{C}}(n,\bar{h},p)  =  \frac{P_S + p  \Pr(h(t) \geq \bar{h}) }{ 1 + n \eta_C(\bar{h},p)}. \label{eq:costONOFFCSI}
\end{equation}
for any $n \in \Zo,\bar{h} \in \mathcal{H}, p \in [0,P_{\max}]$, and for channel inversion policies is given by
\begin{equation}
  J_{\mathrm{I}}(n,\bar{h},\kappa)  =  \frac{P_S+ \sum_{h \in \mathcal{H}, h \geq \bar{h}} \mathcal{P}(h) \rho(h) }{ 1 + n \eta_I(\bar{h},\kappa)}. \label{eq:costinvgen}
\end{equation}
for any $n \in \Zo,\bar{h} \in \mathcal{H}, \kappa \in \Rlo$. \hfill $\Box$\label{lem:gencost}
\end{proposition}
\begin{proof} First, we note that under policy \eqref{eq:TCTPC}, the transmission power $P(\cdot)$ can be seen as a Markov process which depends on the clock state $\tau(\cdot)$ and the CM $h(\cdot)$ and applying Lemma \ref{lem:tau}, given in the Appendix, we have
\begin{equation}
     \Pr(\tau(t) \geq n+1) = \frac{\frac{1}{\eta}}{n + \frac{1}{\eta}}
\end{equation}
When $\tau(t) \geq n$, the expected transmission power can be evaluated as $p \Pr(h(t)\geq \bar{h})$ for constant power policies and $\sum_{h \in \mathcal{H}, h \geq \bar{h}} \mathcal{P}(h) \rho(h) $ for inversion policies as the transmission power is $0$ when $h(t) < \bar{h}$. This allows us to obtain \eqref{eq:costONOFFCSI} and \eqref{eq:costinvgen}.
\end{proof}

Next, we formally state our optimization problems (OP). We have OP$_C$ for constant power policies given by
\begin{equation}\begin{array}{c}
      \text{Minimize}_{n \in \Zo,\bar{h} \in \mathcal{H}, p \in [0,P_{\max}]} J_{\mathrm{C}}(n,\bar{h},p),   \\
   \text{ subject to }  \beta(n,\eta_C(\bar{h},p)) \leq \mu.
\end{array} \label{eq:op1}
\end{equation}
Similarly, we have OP$_I$ for inversion policies given by
\begin{equation}\begin{array}{c}
      \text{Minimize}_{n \in \Zo,\bar{h} \in \mathcal{H},\kappa \in \Rlo} J_{\mathrm{I}}(n,\bar{h},\kappa) ,  \\
   \text{ subject to }  \beta(n,\eta_I(\bar{h},\kappa)) \leq \mu.
\end{array} \label{eq:op2}
\end{equation}

From SA\ref{ass:estimateh}, we have that $\mathcal{H}$ is a discrete and finite set. Therefore, if, for a given $\bar{h} \in \mathcal{H}$, $n$ and $p$ or $\kappa$ can be optimized, OP$_C$ and OP$_I$ can be solved. Next, we show that the set of feasible $n$ such that $\beta(n,\eta) \leq \mu$ is also finite for any $\eta \in [0,1]$. Therefore, if $p$ and $\kappa$ can be optimized for any given $n, \bar{h}$, an exhaustive search over all feasible $n,\bar{h} $ can be performed to solve OP$_C$ and OP$_I$ as they belong to finite sets.

\begin{lemma}
For any given $\mu \in (0,1)$ and any probability of successful transmission $\eta$ in \eqref{eq:muofp}, the set of feasible $n$ satisfying $\beta(n,\eta) \leq \mu$ is finite, and any feasible $n$ is upper-bounded by $N<\infty$, defined by
\begin{equation}
    N: = \max \Big\{n \in \Zo\, | \,(a_Sa_U^n)^{\frac{1}{n+1}} \leq \mu \Big\}. \label{eq:maxN}
\end{equation} \hfill $\Box$
\end{lemma}
\begin{proof}
Notice that for any transmission policy and $\eta \in [0,1]$, $\beta(n,\eta)$ is decreasing in $\eta$ as seen from Lemma \ref{prop:minpar}. We have in view of \eqref{eq:muofp},
\begin{equation}
 \beta(n,1) = (a_Sa_U^n)^{\frac{1}{n+1}},
\end{equation}
for any $n \in \Zo$. \textcolor{black}{Therefore, if $(a_Sa_U^n)^{\frac{1}{n+1}} > \mu $ for some $n \in \Zo$, then $\beta(n,\eta) > \mu $ for all $\eta \in [0,1]$. This condition can be rewritten as $ a_S \left(\frac{a_U}{\mu} \right)^n > \mu  $ and since $a_U>\mu$ from SA\ref{ass:alpha1}, $\frac{a_U}{\mu}$ is greater than $1$ and increasing exponentially in $n$. We can thus define $N < \infty$ according to \eqref{eq:maxN}.}

Then for every $n >N$, we have $\beta(n,\eta) > \mu$ which means that the feasible set of $n$ ensuring $\beta(n,\eta) \leq \mu$ is a subset of $\{0,\dots,N\}$.
\end{proof}

Next, note that OP$_C$ and OP$_I$ are not typically investigated in the wireless communication literature due to the time-based trigger and the constraints added in order to satisfy the control property. Minimizing this cost with respect to $p$ or $\kappa$ is, in general, a difficult problem due to $J_C$ and $J_I$ being non-convex with respect to these variables, and $\eta$ being hard to analyze in general. Therefore, we focus on some special, relevant cases and propose methods, for which we can solve OP$_C$ and OP$_I$ as described below.
\begin{enumerate}
    \item \emph{Pure channel threshold policies}, in which, transmissions can be attempted at any time, provided the CM is greater or equal to $\bar h$. This corresponds to the case where $n=0$ and results in a cost function that is often seen in wireless communications literature \cite{wu2016green}. Although this policy requires channel measurements at all time, when $P_S$ is very small or zero due to purely sensing the channel without sending pilot signals, this policy can perform well. Also, this policy is applicable when ACK packets are unavailable. 
    \item \emph{Pure time triggered policies}, in which transmissions occur whenever $\tau(t) \geq n+1$ irrespective of the actual value of $h(t)$, i.e., $\bar{h}=0$. These policies are relevant when no CM is available at the transmitter and do not consume power for sensing. Naturally, channel thresholds or inversion cannot be applied in this case and so we focus on constant power policies with threshold only on $\tau(t)$.
        \item $\epsilon$-\emph{loss constant power policies}, in which the channel threshold and power are chosen sufficiently large such that communication is almost always successful, i.e., $\psi(\bar{h}p) \geq 1- \epsilon$. These policies are suitable when the feasible transmit power is large but not finely adjustable, leading to almost sure communication success with a sufficiently high channel gain.
    \item \emph{Unsaturated inversion policies}, in which the channel thresholds and the power gain $\kappa$ are chosen such that the channel inversion results in a transmission power smaller than or equal to the maximum power, i.e., $\bar{h} \geq \frac{\kappa}{P_{\max}}$. This allows us to use \eqref{eq:psr2} and thus $\eta$ can be easily evaluated. We also demonstrate in the following how this simplifies tuning $\kappa$.
\end{enumerate}

\subsection{Pure channel threshold policies}\label{sec:purech}

In the policies considered in this subsection, since the decision to transmit or not is determined by the CM $h(t)$ alone, the policy becomes independent of $\tau(t)$. This allows us to relax item (ii) of SA\ref{ass:CSITpd} as this policy can be implemented without any ACK protocol. From Proposition \ref{lem:gencost}, the cost function is given in this case, for any $\bar{h} \in \mathcal{H}, p \in [0,P_{\max}]$, by
\begin{equation}
  J_{\mathrm{C}}(0,\bar{h},p)  = P_S+  \sum_{h \in \mathcal{H}, h \geq \bar{h}} p \rho(h) =P_S+p \Pr(h(t) \geq \bar{h}) \label{eq:costTTCSI0n}
\end{equation}
for constant power policies. The cost function for any $\bar{h} \in \mathcal{H}, \kappa \in \Rlo$ is given by
\begin{equation}
  J_{\mathrm{I}}(0,\bar{h},\kappa)  =  P_S+ \sum_{h \in \mathcal{H}, h \geq \bar{h}} \frac{\kappa}{h} \rho(h)  \label{eq:costTTCSI0nI}
\end{equation}
for inversion policies. Due to these simplified forms of the cost function, we are able to solve OP$_C$ and OP$_I$ as follows.

\begin{proposition}
For any given $\mu \in (0,1), n=0$ and $\bar{h} \in \mathcal{H}$, if $\underline{p}(\mu,0,\bar{h}) \in [0,P_{\max}]$ satisfying \eqref{eq:defup} exists, the solution to OP$_C$ is given by $\underline{p}(\mu,0,\bar{h})$, otherwise OP$_C$ is infeasible. Similarly, if $\underline{\kappa}(\mu,0,\bar{h}) \in \Rlo$ satisfying \eqref{eq:defuk} exists, the optimal $\kappa$ solving OP$_I$ is $\underline{\kappa}(\mu,0,\bar{h})$, otherwise OP$_I$ is infeasible. \label{prop:purech}
\end{proposition}
\begin{proof}
The proof is straightforward upon noticing that $J_{\mathrm{C}}(0,\bar{h},p)$ is increasing in $p$. Recall that the set of feasible $p$ ensuring stochastic stability are such that $p\geq \underline{p}(\mu,0,\bar{h}) $. Therefore, the optimal power minimizing the communication energy cost while satisfying the desired convergence property must be $\underline{p}(\mu,0,\bar{h})$ for constant power policies. Similar arguments hold for inversion polices, completing the proof. 
\end{proof}

 \subsection{Pure time threshold policies}
\label{sec:TTnocsi}

For ease of notation, we define $e(p):= 1- \eta_C(0,p)$, where $\eta_C$ comes from \eqref{eq:etaCP}, the average packet error rate over all possible channel realizations while using a fixed power $p \in [0,P_{\max}]$. We make the following assumption on $e$.
\begin{assumption}
The mapping $e:[0,P_{\max}] \to [0,1]$ is continuous, twice differentiable, initially concave and eventually convex. \label{ass:packetdr2} \hfill $\Box$
\end{assumption}

Assumption \ref{ass:packetdr2} is standard in the wireless communications literature and is observed to hold true for various channel fading models \cite{rodriguez2003analytical, goodman2000power}. Next, from Section \ref{sec:fsets}, we know that the set of feasible powers is given by $[ \underline{p}(\mu,n,0), P_{\max}] $. For the sake of convenience, we will use $\underline{p}_n:= \underline{p}(\mu,n,0)$ throughout this subsection.

The minimum feasible power is given by $\underline p_n $, but using $\underline{p}_n$ does not necessarily imply that the cost \eqref{eq:costperk} is minimized for a given $n \in \Zo$. Indeed, it might be more efficient to use a higher power because we assume that transmissions are attempted until a packet goes through, and using a smaller power would imply a larger number of re-transmissions, thereby potentially increasing the net energy consumed \cite{varma2013energy}, \cite{goodman2000power}, see Section \ref{sec:nr} for an illustration. In the next proposition, we characterize the associated average communication cost. We use the notation $J_{\mathrm{PT}}(p,n):=J_C(n,0,p)$ to denote the cost of a pure time threshold policy with a constant power $p$ for convenience.

\begin{proposition}
Under Assumption \ref{ass:packetdr2}, using a transmission policy based on \eqref{eq:TCTPC} with $\bar{h}=0$, $p \in [0,P_{\max}]$ and $\mathcal{P}(h)=p$, the cost in \eqref{eq:costperk} for all $n \in \Zo$ is given by
\begin{equation}
   J_{\mathrm{PT}}(p,n)=  \frac{ p}{ (1 - e(p))n + 1}. \label{eq:TTcost1}
\end{equation}
Furthermore, the mapping $(p,n) \mapsto J_{\mathrm{PT}}(p,n)$ is
\begin{enumerate}
    \item strictly increasing in $p$ for small $n$,
     \item ``$N$-shaped" for larger values of $n$, i.e., it is initially increasing upto a local maximum, then decreasing to a local minimum and then finally increasing again in $p$. \hfill $\Box$
\end{enumerate}
 \label{th:optPnoCSI} \end{proposition}
 \begin{proof}
 See Appendix \ref{app:p4}.
 \end{proof}

Note that $P_S$ does not appear in the cost \eqref{eq:TTcost1} as the channel is never sensed or estimated for the policies considered in this subsection. We can then exploit Proposition \ref{th:optPnoCSI} to characterize the optimal power $p$ minimizing \eqref{eq:costperk} for a given $n$, such that $(p,n) \in \mathcal{S}$.

\begin{theorem}
Under Assumption \ref{ass:packetdr2}, for any given $n \in \Zo$ with $\beta(n,1-e(P_{\max})) \leq \mu$ and $\bar{h}=0$, OP$_C$ is solved by using the optimal power $p_n^*$, obtained as follows: If a local minimum $p^o_n \in \Rlp$ exists such that $ \frac{\partial J_{\mathrm{PT}}}{\partial p}(p^o_n,n) =0 $ and $ \frac{\partial^2 J_{\mathrm{PT}}}{\partial p^2}(p^o_n,n) >0 $, then $ p_n^* \in \left\{\underline{p}_n,p^o_n , P_{\max} \right\}  $. Otherwise, $ p_n^* = \underline{p}_n$.\hfill $\Box$ \label{th:optp}
\end{theorem}
\begin{proof}
Proposition \ref{th:optPnoCSI} implies $J_{\mathrm{PT}}$ for a given $n$ is either strictly increasing in $p$ or $N$-shaped. In the first case, $p_n^o$ does not exist and so, selecting $\underline{p}_n $ is optimal.

In the second case, $J_{\mathrm{PT}}$ is $N$-shaped in $p$, and it has a single local minimum and is concave for small $p$ and then convex. Since we look at $J_{\mathrm{PT}}(p,n)$ for $p \in [\underline{p}_n, P_{\max}]$, a closed and compact set, the global optimum is either the local minimum or one of the boundary points.  When $p_n^o> P_{\max}$, $J_{\mathrm{PT}}$ may be decreasing or concave in the interval $[\underline{p}_n, P_{\max}]$, which implies that the global minimum is at $\underline{p}_n$ or $P_{\max}$. Otherwise, the optimal power is either $\underline{p}_n $ or the local minimum $p_n^o$.
\end{proof}

Theorem \ref{th:optp} characterizes the optimal power to use for a given $n \in \Zo$. If a local minimum $p^o_n$ for $J_{\mathrm{PT}}$ exists, the optimal power belongs to $\left\{\underline{p}_n,p^o_n , P_{\max} \right\}$. Otherwise, the optimal power is the minimum feasible power $\underline{p}_n$. In practice, the existence of the local minimum $p^o_n \in [\underline{p}_n,P_{\max}]$ for $J_{\mathrm{PT}}$ with a given $n$ can be easily checked by applying a gradient descent initialized at $P_{\max}$. If the gradient descent converges to a point in the interval $[\underline{p}_n,P_{\max}]$, then this point is $p^o_n$, and all elements of the set $\left\{\underline{p}_n,p^o_n , P_{\max} \right\}$ can be tested to find the optimum.

\subsection{$\epsilon$-loss constant power policies}

We focus on policies with $\bar{h},p$ selected such that $\psi(\bar{h}p) \geq 1 - \epsilon$ with $0<\epsilon \ll 1$, i.e., when communicating, the packet is successful with a probability close to $1$. We define $p_\epsilon(\bar{h})$ as the solution to
\begin{equation}
 \psi(\bar{h} p_\epsilon(\bar{h})) = 1 - \epsilon.
\end{equation}

If $p_\epsilon(\bar{h}) \in [0,P_{\max}]$ exists, then $\bar{h}$ is a feasible channel threshold for the $\epsilon$-sure constant power policy and all $p \in [p_\epsilon(\bar{h}),P_{\max}] $ are feasible. We make use of the following result to provide optimality conditions on the communication cost.

\begin{proposition}
For any $\bar{h}$, $p$, we have
\begin{equation}
 \Pr(h(t) \geq \bar{h}) \psi(\bar{h}p) \leq   \eta_C(\bar{h},p) \leq \Pr(h(t) \geq \bar{h}).
\end{equation} \label{prop:etacbounds} \hfill $\Box$
\end{proposition}
\begin{proof}
Recall that $\psi$ is a strictly \textcolor{black}{increasing} function according to SA\ref{ass:CSITpd}. Therefore, the term $\psi(ph) \rho(h)$ in the summation expression in $\eta_C$, as seen from Lemma \ref{lem:etaCP} is lower and upper bounded by $\psi(\bar{h}p) \rho(h)$ and $\rho(h)$ respectively. Note that since $\rho(h)=\Pr(h(t)=h)$, we have $\Pr(h(t) \geq \bar{h}) = \sum_{h \in \mathcal{H}, h \geq \bar{h}} \rho(h) $. 
\end{proof}

For a given $n, \bar{h}$, from Proposition \ref{prop:etacbounds}, the desired control properties \eqref{eq:stmu} are ensured for any $p \in [p_\epsilon(\bar{h}),P_{\max}] $ as long as
\begin{equation}
    \beta(n, \Pr(h(t) \geq \bar{h}) (1-\epsilon)) \geq \mu
\label{eq:ssapprox}
\end{equation}

We approximate $\eta_C(\bar{h},p)$ as follows to find the optimal policy, which is justified by $\epsilon$ being very small, and thus the difference of the approximation to the exact value becomes of order $\epsilon$.
\begin{assumption}
For all $p \in [p_\epsilon(\bar{h}),P_{\max}] $, $\eta_C(\bar{h},p) =\Pr(h(t) \geq \bar{h}) (1-\epsilon) $. \hfill $\Box$ \label{ass:almostsure}
\end{assumption}

We now search for the optimal $p$ solving OP$_C$ under Assumption \ref{ass:almostsure}.

\begin{theorem}
Under Assumption \ref{ass:almostsure}, for a given $n \in \Zo$, if there exists $\bar{h}^* \in \mathcal{H}$ such that
\begin{equation}
    \beta(n, \Pr(h(t) \geq \bar{h}^*) (1-\epsilon)) =\mu, \label{eq:minimalh}
\end{equation}
then $p_\epsilon(\bar{h})$ is the optimal power for any $\bar{h} \leq \bar{h}^*$ solving OP$_C$. If no such $\bar{h}^*$ exists, then the $\epsilon$-loss constant power policy is infeasible.  \hfill $\Box$
\end{theorem}

\begin{proof}
Note that $J_C(n,\bar{h},p)=( P_S + p  \Pr(h(t) \geq \bar{h}) (1+ n \Pr(h(t) \geq \bar{h})(1-\epsilon))^{-1}$ from \eqref{eq:costONOFFCSI} as we consider $\eta_C(\bar{h},p)=1- \epsilon$ under Assumption \ref{ass:almostsure}. Thus, $J_C(n,\bar{h},p)$ is strictly increasing in $p$. Therefore, taking the smallest power results in the smallest cost, which is achieved by $p_\epsilon(\bar{h})$.
\end{proof}

These policies are well suited for communication systems where the power $p$ cannot be fine tuned, but can only be set at certain levels, such as $\{0,P_{\max}\}$. The condition \eqref{eq:minimalh} is easily verifiable as $n$ and $\bar{h}$ belong to finite discrete sets and an exhaustively search can be applied to find all feasible values. 

\subsection{Unsaturated inversion policies}

In this subsection, we focus on channel inversion policies with $\bar{h} \geq \frac{\kappa}{P_{\max}} $, which implies that $\frac{\kappa}{h(t)} \leq P_{\max}$ for all $t$ when $h(t) \geq \bar{h}$. Consequently, from Lemma \ref{lem:IP}, we have that $\eta_I(\bar{h},\kappa)= \Pr(h(t) \geq \bar{h}) \psi(\kappa)$. For a given $\mu \in (0,1), n \in \Zo$ and $\bar{h} \in \mathcal{H}$, the smallest $\kappa$ satisfying the stability and convergence property is given by $\underline{\kappa}(\mu,n,\bar{h})$ from Section \ref{sec:fsets}. Observe that if $\underline{\kappa}(\mu,n,\bar{h})< \bar{h}P_{\max}  $, then the unsaturated inversion policy becomes feasible for any $\kappa \in [\underline{\kappa}(\mu,n,\bar{h}), \bar{h}P_{\max}] $. This allows us to derive the following theorem.

\begin{theorem} \label{th:unsat}
For any given $n \in \Zo, \mu \in (0,1)$ and $\bar{h} \in \mathcal{H}$ such that $\underline{\kappa}(\mu,n,\bar{h})\leq P_{\max} \bar{h} $, the optimal gain solving OP$_I$ is given by $\kappa^*$ obtained as follows: If a local minimum $\kappa^o \in \Rlp$ exists such that $ \frac{\partial J_{I}}{\partial \kappa}(n,\bar{h},\kappa^o) =0 $ and $ \frac{\partial^2 J_{I}}{\partial \kappa^2}(n,\bar{h},\kappa^o) >0 $, then the optimal gain $ \kappa^* \in \left\{\underline{\kappa}(\mu,n,\bar{h}),\kappa^o , P_{\max} \bar{h} \right\}  $. Otherwise, the optimal gain $ \kappa^* =\underline{\kappa}(\mu,n,\bar{h})$. \hfill $\Box$
\end{theorem}

\begin{proof}
The mathematical properties of $ J_{\mathrm{I}}(n,\bar{h},\kappa) $ are identical to the properties of $J_{\mathrm{PT}}(p,n)$ for a given value of $\bar{h}$ with $\kappa$ being replaced by $p$ for unsaturated policies due to the first case of \eqref{eq:psr2}. We can thus follow the proof of Theorem \ref{th:optp} to prove this result.  
For unsaturated inversion policies, the cost function \eqref{eq:costinvgen}, which can be rewritten as
\begin{equation}
  J_{\mathrm{I}}(n,\bar{h},\kappa)  = \frac{P_S+\kappa \sum_{h \in \mathcal{H}, h \geq \bar{h}} \frac{ \rho(h)}{h} }{1+n \Pr(h(t) \geq \bar{h}) \psi(\kappa)} , \label{eq:costunsat}
\end{equation}
is N-shaped w.r.t $\kappa$. Consequently the optimal $\kappa$ for any given $n,\mu$ and $\bar{h}$ can be found using the same method explained in the proof of Theorem \ref{th:optp}.
\end{proof}

\section{Conditions ensuring SA\ref{ass:alpha1}}\label{sect:assumption}

Before illustrating the results of Section \ref{sec:comm} on an example, we demonstrate how to systematically satisfy SA\ref{ass:alpha1} for the case of a linear time-invariant  plant and controller. Afterwards, we consider again the nonlinear setting, and propose conditions to guarantee SA\ref{ass:alpha1} when the strategy used to generate $\hat y$ is based on zeroing and zero-order-hold, respectively. 

\subsection{Linear time-invariant systems} \label{sect:assumptionlinear}
We consider the case in which the plant \eqref{eq:mainsys} is linear and time-invariant, i.e.,
\begin{equation}
\begin{array}{rllll}
f_p(x_p,u) & = & A_p x_p+B_p u\\
g_p(x_p) & = & C_p x_p,
\end{array}
\label{eq:mainsyslin}
\end{equation}
where the pairs $(A_p,B_p)$ and $(A_p,C_p)$ are assumed to be stabilizable and detectable, respectively. Here, we can design an output-feedback stabilizing controller for system (\ref{eq:mainsyslin}) as
\begin{equation}
\begin{array}{rllll}
f_c(x_c,y) & = & A_c x_c + B_c y\\ 
g_c(x_c,y)& = & C_c x_c + D_c y,
\end{array}
\label{eq:maincontrollin}
\end{equation}
in the sense that the closed-loop state matrix $\mathcal{A}_0:= \left[ \begin{array}{cc}A_p +  B_p D_c C_p& B_pC_c \\ B_c C_p  & A_c  \end{array} \right]$ is Schur. Between two successive successful transmission instants, $\hat{y}$ is held using a linear holding function $\hat{f}(y)  =  C_gy$ for some $C_g \in \Rl^{s_y \times s_y}$ and any $y \in \Rl^{s_y}$. 

Let $\chi:=(x_p,x_c,\hat{y})$ as in Section \ref{sec:sysm}, and we obtain \eqref{eq:mainsys2} with
\begin{equation}
\begin{array}{llr}
f_S(\chi) & = & \mathcal{A}_S  \chi\\ 
f_U(\chi)& = &\mathcal{A}_U \chi, \label{eq:lin1}
\end{array}
\end{equation}
where \begin{equation}
\mathcal{A}_S :=  \left[ \begin{array}{ccc}A_p +  B_p D_c C_p& B_pC_c & 0\\ B_c C_p  & A_c & 0 \\ C_g C_p & 0 & 0 \end{array} \right],  \end{equation}  
and
\begin{equation}
\mathcal{A}_U :=  \left[ \begin{array}{ccc}A_p & B_p C_c  & B_p  D_c  \\ 0  & A_c & B_c \\ 0 & 0 & C_g  \end{array} \right]. \label{eq:lin3}
\end{equation} 

The next proposition ensures that Assumption 1 always holds for system \eqref{eq:mainsys2} with \eqref{eq:lin1}-\eqref{eq:lin3}.

\begin{proposition} \label{prop:linearcase}
SA\ref{ass:alpha1} holds for system \eqref{eq:mainsys2} with \eqref{eq:lin1}-\eqref{eq:lin3} by taking $V(\chi)= \chi^{\top} P \chi$ for any $\chi \in \Rl^{s_\chi}$ with $P \in \Rl^{s_\chi \times s_\chi}$ symmetric, positive definite and such that
\begin{equation}
\begin{array}{llr}
 \mathcal{A}_S^{\top}  P  \mathcal{A}_S  & \leq & a_S  P\\ 
 \mathcal{A}_U^{\top}  P  \mathcal{A}_U  & \leq & a_U  P\\ 
\end{array} \label{eq:lmilin}
\end{equation}
with $a_S \in [0,1)$, $a_U >a_S$, $\overline{\alpha}(s)=\lambda_{\max}(P)s^2$ and $\underline{\alpha}(s)=\lambda_{\min}(P)s^2$.  \hfill $\Box$
\end{proposition}
\begin{proof}
Since $ \mathcal{A}_0$ is Schur, there exists $P_0 \in \Rl^{s_x \times s_x}$ which is symmetric, positive definite and such that
\begin{equation}
 \mathcal{A}_0^{\top}  P_0  \mathcal{A}_0   \leq  a_0  P_0 \label{eq:linearorg} \end{equation}
for some $a_0 \in [0,1)$. Let $P:= \left[ \begin{array}{cc} P_0 & 0 \\ 0 & \epsilon  \end{array} \right]$ with $\epsilon>0$, which is thus symmetric and positive definite. We have that
 \begin{equation}
     \mathcal{A}_S^{\top} P  \mathcal{A}_S =\left[ \begin{array}{cc}  \mathcal{A}_0^{\top} P_0  \mathcal{A}_0 + \epsilon C_g^{\top} C_p^{\top} C_g C_p  & 0 \\ 0  & 0 \end{array} \right] . \label{eq:linearint1}
      \end{equation}
Let $\chi \in \Rl^{s_\chi}$, we have
\begin{equation} \begin{array}{lll}
    V(\mathcal{A}_S \chi)   & =&  \chi^{\top}   \mathcal{A}_S^{\top} P  \mathcal{A}_S \chi \\
     & =& x^{\top} \mathcal{A}_0^{\top} P_0  \mathcal{A}_0 x + \epsilon x_p^{\top} C_g^{\top} C_p^{\top} C_g C_p x_p
\end{array}
\end{equation}
using \eqref{eq:linearint1} and since $x=(x_p,x_c)$. In view of \eqref{eq:linearorg}, we have 
\begin{equation} \begin{array}{lll}
    V(\mathcal{A}_S \chi)   & \leq & a_0 x^{\top} P_0 x + \epsilon |C_g C_p|^2 |x_p|^2  \\
     & \leq &a_S x^{\top} P_0 x \leq a_S V(\chi)
\end{array} \label{eq47}
\end{equation}
with $a_S: = a_0+ \epsilon |C_g C_p|^2 \lambda_{\min}(P_0)^{-1}$. \textcolor{black}{In the first line of \eqref{eq47}, we apply \eqref{eq:linearorg} to bound the term $x^{\top} \mathcal{A}_0^{\top} P_0  \mathcal{A}_0 x$. Since $x=(x_p,x_c)$, we can always find some $a_S>0$ such that the first line of \eqref{eq47} is bounded by the second line. Next, as the second term depends on $\epsilon$, by taking $\epsilon$ sufficiently small, we can always find $a_S \in (a_0,1)$ as $a_0<1$.} On the other hand, we have for any $\chi \in \Rl^{s_\chi}$
\begin{equation} \begin{array}{lll}
    V(\mathcal{A}_U \chi)   & =&  \chi^{\top}   \mathcal{A}_U^{\top} P  \mathcal{A}_U \chi \\
     &  \leq & |\mathcal{A}_U^{\top} P  \mathcal{A}_U|  |\chi|^2 \leq a_U V(\chi)
\end{array}
\end{equation}
with $a_U \geq |\mathcal{A}_U^{\top} P  \mathcal{A}_U| \lambda_{\min}(P)^{-1}$. Since all these terms are positive, we can always find some $a_U>a_S$ thus satisfying \eqref{eq:lmilin}, and consequently SA\ref{ass:alpha1}.
\end{proof}

Conditions \eqref{eq:lmilin} are linear matrix inequalities (LMIs) with respect to unknowns $P$, $a_S \in (0,1)$ and $a_U >a_S$, which always have a solution in view of Proposition \ref{prop:linearcase}. 
A desired pair $(a_U$, $a_S)$ less conservative than the one obtained following the proof can be tested for feasibility numerically using an LMI solver.

\subsection{Zeroing strategy}

We return to a general plant and controller models as in \eqref{eq:mainsys} and \eqref{eq:maincontrol}, and we focus on zeroing strategies to generate $\hat{y}$, i.e., $\hat{f}=0$. We suppose that controller \eqref{eq:maincontrol} has been designed such that the following properties hold.

\begin{assumption}\label{ass:zeroing} There exist $W:\mathbb{R}^{s_p+s_c}\to\mathbb{R}$ continuous, $\underline\alpha_W,\overline\alpha_W\in\mathcal{K}_\infty$, $a_{W,1}\in(0,1)$ and $a_{W,0} >0$ such that, for any $(x_p,x_c)\in\mathbb{R}^{s_p+s_c}$:
    \begin{enumerate}
        \item[(i)] \hspace{-0.2cm} $\underline\alpha_W(|(x_p,x_c)|)\leq W(x_p,x_c)\leq\overline\alpha_W(|(x_p,x_c)|)$;
        \item[(ii)]\hspace{-0.2cm} $W(f_p(x_p,g_c(x_c,g_p(x_p)),f_c(x_c,g_p(x_p))) \leq a_{W,1} W(x)$;
        \item[(iii)]\hspace{-0.2cm} $W\left(f_p(x_p,g_c(x_c,0),f_c(x_c,0)\right) \leq a_{W,0} W(x)$. \hfill $\Box$
    \end{enumerate}
\end{assumption}

Items (i)-(ii) of Assumption \ref{ass:zeroing} are equivalent to the fact that the origin of \eqref{eq:mainsys}-\eqref{eq:maincontrol} is UGAS when $f_p$, $f_c$, $g_p$ and $g_c$ are continuous, see \cite{Jiang-Wang-scl02(converse)}. Item (iii), on the other hand, is an exponential growth condition on $W$ when a transmission fails and  $\hat f$ in \eqref{eq:holding} is the zero function. The next proposition ensures the satisfaction of SA\ref{ass:alpha1}.

\begin{proposition}\label{prop:ass-zeroing} Suppose Assumption \ref{ass:zeroing} holds, then SA\ref{ass:alpha1} is verified with $V:\chi\mapsto W(x_p,x_c)+ |\hat y|$,  $a_S=a_{W,1}$, $a_U=a_{W,0}$,  $\underline\alpha(s)=\min\{\underline\alpha_W(s/2),s/2\}$ and $\overline{\alpha}(s)=\overline\alpha_W(s)+s$ for any $s\geq 0$. \hfill $\Box$ 
\end{proposition}

\emph{Proof:} Let $\chi\in\mathbb{R}^{s_\chi}$, $V(\chi)\leq\overline\alpha_W(|(x_p,x_c)|)+|\hat y|$ in view of item (i) of Assumption \ref{ass:zeroing}, from which we derive that $V(\chi)\leq\overline\alpha_W(|\chi|)$ with $\overline\alpha_W$ given in Proposition \ref{prop:ass-zeroing}. We obtain the lower-bound on $V$ by invoking \cite[Remark 2.3]{Laila-Nesic-cdc-02(small-gain)}. On the other hand, in view of item (ii) of Assumption \ref{ass:zeroing}, $V(f_S(\chi)) =W\left(f_p(x_p,g_c(x_c,g_p(x_p)),f_c(x_c,g_p(x_p))\right) \leq  a_S W(x) \leq a_{W,1} V(\chi)$. We similarly derive from item (iii) of Assumption \ref{ass:zeroing} that $V(f_U(\chi))\leq a_{W,0} V(\chi)$, which concludes the proof. \hfill $\blacksquare$


\subsection{Zero-order-hold strategy}

When zero-order-hold devices are used to generate $\hat y$, we introduce Assumption \ref{ass:zoh} to conclude about the satisfaction of SA\ref{ass:alpha1}. 

\begin{assumption} \label{ass:zoh} There exist $W:\mathbb{R}^{s_p+s_c}\to\mathbb{R}$ continuous, $\underline a_W,\overline a_W>0$, $a_{W,1}\in(0,1)$ and $a_{W,0},b_0\geq 0$ such that, for any $\chi\in\mathbb{R}^{s_\chi}$:
    \begin{enumerate}
        \item[(i)]\hspace{-0.2cm} $\underline a_W |(x_p,x_c)|^2\leq W(x_p,x_c)\leq\overline a_W |(x_p,x_c)|^2$;
        \item[(ii)]\hspace{-0.2cm} $W\left(f_p(x_p,g_c(x_c,g_p(x_p)),f_c(x_c,g_p(x_p))\right) \leq a_{W,1} W(x)$;
        \item[(iii)] \hspace{-0.2cm}$W\left(f_p(x_p,g_c(x_c,\hat y),f_c(x_c,\hat y)\right) \leq a_{W,0} W(x) + b_0 |\hat y|^2$. \hfill $\Box$
    \end{enumerate}
\end{assumption}

Items (i)-(ii) of Assumption \ref{ass:zoh} are equivalent to the fact that the origin of \eqref{eq:mainsys}-\eqref{eq:maincontrol} is uniformly globally exponentially stable under conditions as mentioned after Assumption \ref{ass:zeroing}. Item (iii) is an exponential growth condition on $W$ when a transmission fails, which involves $\hat y$ this time because of the use of a zero-order-hold strategy. 

We also require the output map to be linearly bounded.

\begin{assumption}\label{ass:zoh-output-map}
There exist $c\geq 0$ such that $|g_p(x_p)|\leq c|x_p|$  for any $x_p\in \mathbb{R} ^{s_p}$. \hfill $\Box$
\end{assumption}

Assumption \ref{ass:zoh-output-map} is verified when $y=C_p x_p$ with $C_p$ a real matrix (for instance) in which case $c=||C_p||$. The next proposition ensures the satisfaction of SA\ref{ass:alpha1}.

\begin{proposition}\label{prop:ass-zoh} Suppose Assumptions \ref{ass:zoh}-\ref{ass:zoh-output-map} hold, then SA\ref{ass:alpha1} is verified with $V:\chi\mapsto W(x_p,x_c)+ \nu|\hat y|^2$ for some $\nu\in \displaystyle \left(0,(1-a_{W,1})\frac{\underline a_W}{c^{2}}\right)$, $a_S=a_{W,1}+\nu c^2/\underline{a}_W$, $a_U=\max\{a_{W,0}, b_0/\nu + 1\}$ given in Assumption \ref{ass:zeroing},  $\underline\alpha(s)=\min\{\underline\alpha_U(s/2),\nu (s/2)^2\}$ and $\overline{\alpha}(s)=\overline\alpha_U(s)+\nu s^2$ for any $s\geq 0$. \hfill $\Box$ 
\end{proposition}

\emph{Proof:} The proof of (\ref{eq:ass-sandwich-bounds}) follows similar lines as in the proof of Proposition \ref{prop:ass-zeroing}. Let $\chi\in\mathbb{R}^{s_\chi}$. In view of item (ii) of Assumption \ref{ass:zoh}, $V(f_S(\chi)) =W\left(f_p(x_p,g_c(x_c,g_p(x_p)),f_c(x_c,g_p(x_p))\right) + \nu |g_p(x_p)|^2\leq  a_{W,1} W(x) + \nu |g_p(x_p)|^2$. According to Assumption \ref{ass:zoh-output-map}, $|g_p(x_p)|^2\leq c^2 |x_p|^2$, and, in view of item (i) of Assumption \ref{ass:zoh}, $|g_p(x_p)|^2\leq c^2/\underline{a}_W W(x_p,x_c) \leq c^2/\underline{a}_W V(\chi)$. Consequently, $V(f_S(\chi))\leq a_{W,1} W(x) + \nu c^2/\underline{a}_W V(\chi) \leq \left(a_{W,1}+\nu c^2/\underline{a}_W\right) V(\chi)= a_S V(\chi)$ and $a_S\in(0,1)$ in view of the definition of $\nu$.

On the other hand, $V(f_U(\chi)) \leq  a_{W,0} W(x) + b_0 |\hat y|^2 + \nu |\hat y|^2$ in view of (\ref{eq:deffu}) and item (iii) of Assumption \ref{ass:zoh}. Hence, $V(f_U(\chi)) \leq  \max\{a_{W,0}, b_0/\nu + 1\}(W(x)+\nu |\hat y|^2)=\max\{a_{W,0}, b_0/\nu + 1\} V(\chi) = a_U V(\chi)$. We have proved that the conditions in SA\ref{ass:alpha1} are verified. \hfill $\blacksquare$

\section{Numerical examples}\label{sec:nr}
\subsection{Single link robot arm and its controller}
We illustrate the results of Section III on a single link robot arm, whose model is obtained by discretizing the continuous-time system using an Euler method with sampling period of $10^{-3}$ seconds. System \eqref{eq:mainsys} with plant state $x_p = (x_1, x_2) \in \mathbb{R}^2$ is given by
\begin{equation}
\left( \begin{array}{c} x_1(t+1) \\ x_2(t+1)
\end{array}\right)=  \left(   \begin{array}{cc}
   x_1(t)+ 10^{-3} x_2(t) \\
    x_2(t) + 10^{-3}(\sin(x_1(t)) + u(t) )
     \end{array} \right). \label{eq:numsys}
\end{equation}
The control \eqref{eq:maincontrol} is given by strategy $u= -\sin(x_1) -25x_1 - 10x_2 $ and we use zero-order-holds to implement it.

SA\ref{ass:alpha1} is verified with $V(\chi) \mapsto \chi^{\top} P \chi$, $a_S=0.98$ and $a_U=1.0009$ where $$P=\left( \begin{array}{cccc}
    0.0384 &  -0.0019 &  -0.0336  &  0.0031\\
   -0.0019 &   0.0015  &  0.0033  & -0.0008\\
   -0.0336  &  0.0033   & 0.0341 &  -0.0032\\
    0.0031 &  -0.0008  & -0.0032  &  0.0009\\
\end{array}  \right).$$
\subsection{Communication settings}
We fix $\mu=0.999$ and apply Proposition \ref{prop:etastab} to obtain the minimum $\eta$ required to ensure the desired stability property for $n\in \{0,1,\dots,10\}$ as $\{ 0.092,  0.097,  0.11 ,0.12, 0.14 ,0.16,  0.19,0.24,  0.32, 0.47, 0.9\}$ respectively. 

We first study the case where CMs are available at the transmitter and consider the packet success rate to be given by $\phi(\gamma)= (0.5+0.5 \mathrm{erf}(\sqrt{\gamma}))^{32}$, which corresponds to the probability that every single bit in a packet of 32 bits is decoded correctly under a white Gaussian noise. For the communication channel model, we consider that the CMs are the quantization of a Rayleigh slow-fading channel with a probability distribution function $0.5\exp(-0.5 \omega) $. We take $\mathcal{H}:=\{0,0.05,\dots,5\}$ in SA\ref{ass:estimateh} and thus obtain $$\rho(h)= \int_{\omega=h}^{h+0.05} 0.5\exp(-0.5 \omega) d\omega, $$ for all $h \in \mathcal{H} \setminus \{5\}$ and $$\rho(5)= \int_{\omega=5}^{\infty} 0.5\exp(-0.5 \omega) d\omega.$$

\subsection{Pure channel thresholds}

We fix $P_S=0$ and $N=0$ to obtain Figure \ref{fig:Jvsh0} which depicts the cost function using a constant power level and a channel inversion policy for feasible values of $\bar{h}$ and optimal $p$ and $\kappa$ respectively based on Proposition~\ref{prop:purech}. We obtain a minimum cost of $0.12$ with channel inversion and $0.16$ for the constant power policy with $\bar{h}=2.2$. However, note that measuring the channel is not always possible and when it is possible, $P_S$ is added to the final cost expression for $n=0$ as seen in \eqref{eq:costTTCSI0nI}. Therefore, if $P_S >0.04$, the constant power policy will outperform the inversion policy. 

\begin{figure}[ht]
  \vspace*{1em}
 \hspace{-3mm}\includegraphics[width=0.53\textwidth,trim={0 9cm 0 9cm},clip]{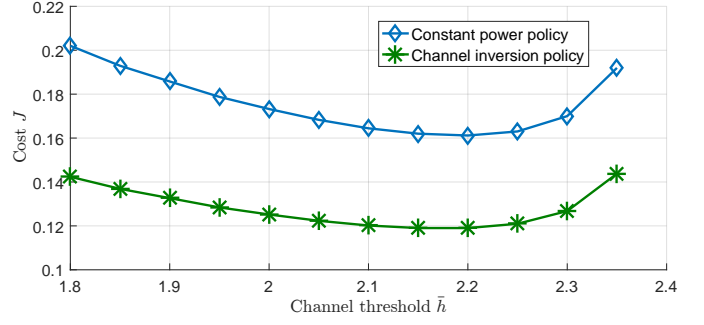}
  \caption{The cost $J$ plotted for the optimal 
 $p$ (under a constant power policy) or $\kappa$ (under a channel inversion policy) for feasible values of $\bar{h} \in \{1.8,1.9,\dots,3\}$ with $n=0$.}\label{fig:Jvsh0}
\end{figure}

\subsection{Pure time thresholds}
Next, we look at the case where CMs are not available and consider parameters such that $e(p) \mapsto 1 - \exp(-1/p)$, which verifies SA\ref{ass:CSITpd}, see \cite{ozarow1994information} for details. This results in $\underline{p}_{10}=9.5$ with $P_{\max}=10$ using \eqref{eq:muofp} and Theorem \ref{th:sstab}. The value of $P_S$ is irrelevant in this case as we never measure the channel.

In Figure \ref{fig:pvn}, we plot the optimal power $p_n^*$ minimizing $J_{\mathrm{PT}}(p,n)$ for $n \in \{0,\dots,19\}$ and compare it with the required power $\underline{p}_n$ to ensure the convergence property \eqref{eq:stmu}. We note that $\underline{p}_{n}$ is not always the optimal power as explained in Theorem \ref{th:optp}. In Figure \ref{fig:noCSI}, we plot the average power consumed $J_{\mathrm{PT}}(p_n^*,n)$ with respect to feasible values of $n$ for given values of $\mu$, when using the optimal power $p_n^*$ as defined in Theorem \ref{th:optp}. We note that using the largest values of feasible $n$ results in a higher communication cost because while the frequency of communications decreases, the power required to stabilize the system also increases with $n$. The optimal $n$ for $\mu \in \{0.995,0.999,0.9999\}$ can be observed to be $1,8, 18$ respectively. We observe that a smaller $\mu$ demands more frequent communication, leading to a higher communication cost, but ensures a faster guaranteed convergence of the Lyapunov function along the solutions to the WCNS.

\begin{figure}[ht]
  \vspace*{1em}
 \centering
 \hspace{-3mm}\includegraphics[width=0.53\textwidth,trim={0 9cm 0 9cm},clip]{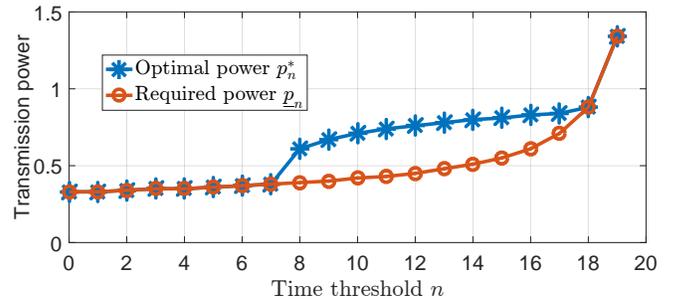}
 \caption{We plot the optimal power $p^*_n$ minimizing $J_{\mathrm{PT}}(p,n)$, and the required power $\underline{p}_n$ for given values of $n$.}\label{fig:pvn}
 \vspace{0.3cm} 
\end{figure}

\begin{figure}[ht]
  \vspace*{1em}
 \hspace{-3mm}\includegraphics[width=0.53\textwidth,trim={0 9cm 0 9cm},clip]{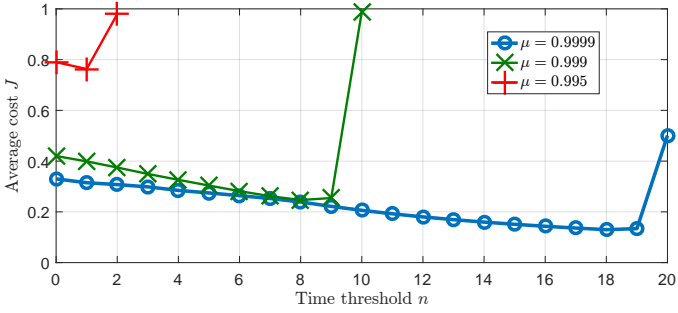}
  \caption{Average power consumed $J_{\mathrm{PT}}(p^*_n,n)$ using the optimal transmission power $p^*_n$ for the given $n$ and $\mu$.}\label{fig:noCSI}
\end{figure}

\subsection{Unsaturated inversion and $\epsilon$-loss policies}
Next, in Figure \ref{fig:Jvsh9}, we plot the cost function using an $\epsilon$-loss constant power policies and unsaturated channel inversion policies for $n=9$, $P_S=0$ and for some feasible values of $\bar{h}$. Here, we apply the optimal values of $p$ for constant power policies according to Proposition \ref{prop:etacbounds} with $\epsilon=0.99$, and the optimal $\kappa$ for inversion policies based on Theorem \ref{th:unsat}. We note that the communication energy cost $J$ is minimized for $\bar{h}=0.65$ and using a channel inversion policy results in a cost of $0.29$ compared to a cost of $0.57$ using a constant power policy, i.e., the cost is almost halved. 

We observe that increasing $n$ is not always good, despite seemingly transmitting less often. This is because of the larger channel threshold that is feasible with a small $n$. This property is demonstrated in Figure \ref{fig:Jvsn}, where we plot the minimum cost achievable using unsaturated inversion policies (by optimizing $\bar{h},\kappa$) for all feasible values of $n$ and various values of $P_S$ using Theorem~\ref{th:unsat}. However, it is important to note that this behavior occurs due to the distribution of the channel and different distributions may change the results presented here. We also observe that when $P_S$ is large, using a larger time threshold $n$ is more efficient. 

\begin{figure}[ht]
  \vspace*{1em}
 \hspace{-3mm}\includegraphics[width=0.53\textwidth,trim={0 9cm 0 9cm},clip]{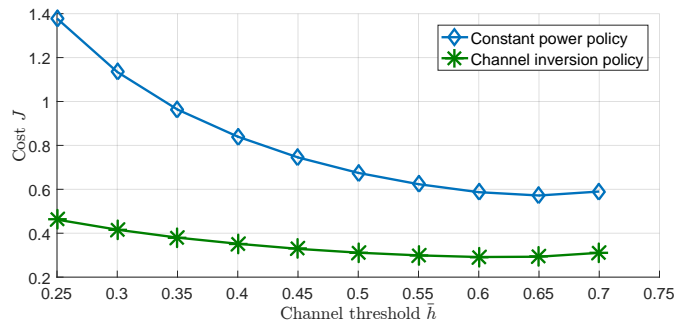}
  \caption{The cost $J$ plotted for the optimal 
 $p$ (under a constant power policy) or $\kappa$ (under a channel inversion policy) for feasible values of $\bar{h}$ with $n=9$.}\label{fig:Jvsh9}
\end{figure}

\begin{figure}[ht]
  \vspace*{1em}
 \hspace{-3mm}\includegraphics[width=0.53\textwidth,trim={0 9cm 0 9cm},clip]{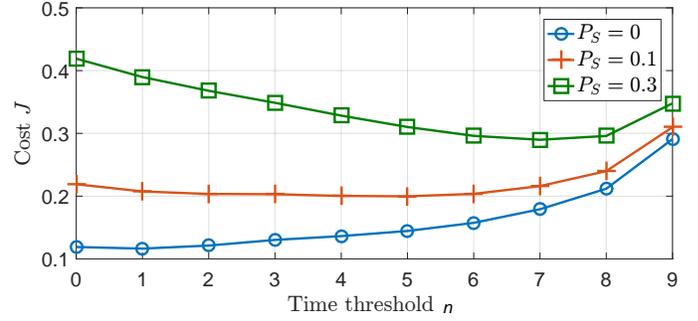}
  \caption{The cost $J$ plotted for the feasible values of $n$ and various values of the sensing power $P_S$, with the optimal parameters $\bar{h},p$ and $\kappa$ being selected for a channel inversion policy.}\label{fig:Jvsn}
\end{figure}

\color{black}
\subsection{Control performance vs communication cost}
\label{sec:tradeoff}
In this subsection, we perform an actual simulation of \eqref{eq:numsys} to study the trade-off between communication and control performance. While $\mu$ gives a guaranteed property on the convergence speed, its value may be subject to conservatism compared to the actual speed. We thus compare the expected time steps (averaged for $10^5$ simulations) for $|x|^2$ to reach a ball of radius $10^{-6}$ with a random initialization on $x(0)$ satisfying $|x(0)|=1$. In Fig. \ref{fig:timevsmu}, we plot the results of this numerical experiment for four communication policies: the baseline which uses $P(t)=P_{\max}$ for all $t$, and the remaining three being the optimal unsaturated inversion policy for $\mu \in \{0.99,0.995,0.999\}$. 
Naturally, a higher $\mu$ implies a higher convergence time. Surprisingly, we discover that the control performance in simulation deteriorates by a very small amount in actual simulations compared to the theoretical bound which scales with $-1/\log(\mu)$. 
\begin{figure}[ht]
  \vspace*{1em}
 \hspace{-3mm}\includegraphics[width=0.5\textwidth,trim={2cm 9cm 2cm 9cm},clip]{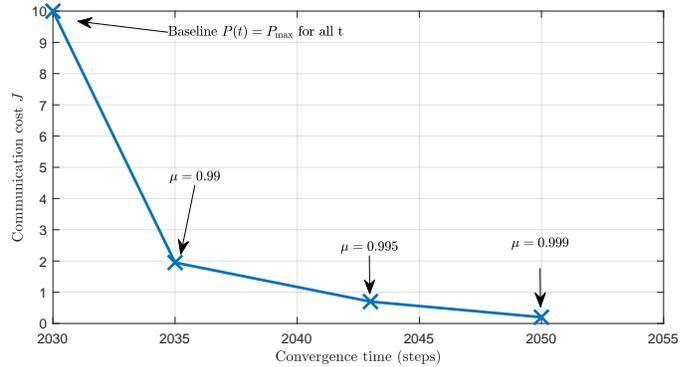}
  \caption{We plot $J$ on the Y axis and on the X axis, the expected time steps for $x$ to hit a ball of radius $10^{-6}$ while starting from a random point in the circle of radius $1$ in our simulations. The communication cost is optimized for a certain $\mu$ as indicated in the figure. }\label{fig:timevsmu}
\end{figure}
\color{black}
\section{Conclusions}
\label{sec:concl}
We have proposed a framework to design a class of energy-efficient transmission power policies for nonlinear WNCS. The main objective of this work is to minimize the average transmission power while maintaining the stability of the WCNS in a stochastic sense. We provide expressions to compute the optimal transmission power for control-relevant performance criteria under the proposed policy based on time and channel thresholds. Numerical simulations show that when the power required to sense the channel is ignored, a pure channel threshold policy can be optimal. However, when the power required to sense the channel is also accounted for, a suitable time threshold will significantly reduce the average communication cost. 

\appendix
It is convenient to define the function $F:\Rl ^{s_\chi \times\Zp} \to \Rl^{s_\chi}$ with the following recursion, for any $\chi \in \mathbb{R}^{s_{\chi}}$ and $\ell\in \Zp $,
\begin{equation}
    F(\chi,\ell):=\left\{ \begin{array}{ll}
    f_S(\chi)     & \text{for } \ell=1  \\
   f_U(F(\chi,\ell-1))     & \text{otherwise.}
    \end{array} \right. \label{eq:mainsys3}
\end{equation}
This allows us to write the dynamics between successful communication instants as $\chi(t_k + \ell) = F(\chi(t_k),\ell_k)$ for all $t_k \in \mathcal{T}$ and $\ell_k \in \{1,\dots,t_{k+1}-t_k\}$.  

We next provide a lemma on the evolution of $\tau(t)$ along solutions to \eqref{eq:mainsys3b} where $\eta$ denotes the packet success rate when $\tau(t) \geq n+1$.
 
\begin{lemma}
For any $\eta \in (0,1)$, the sequence $\tau(1),\tau(2),\dots$ is a Markov chain and $\tau(t)$ converges exponentially to the following stationary distribution:
\begin{equation}
\Pr(\tau(t)=j)= \frac{\eta}{n \eta +1},
\end{equation}
for all $j \in \{1,2,\dots,n\}$ and
\begin{equation}
    \Pr(\tau(t) \geq n+1) = \frac{1 - \eta}{n\eta +1}.
\end{equation} \label{lem:tau} \hfill $\Box$
\end{lemma}
\begin{proof} \textcolor{black}{Equation \eqref{eq:mainsys3b} allows us to evaluate $\Pr(\tau(t+1)| \tau(t))$. Since $P(t)=0$ when $\tau(t) \leq n$, and as $\psi(0)=0$ from SA\ref{ass:CSITpd}, we have that
\begin{equation}
    \Pr\Big(\tau(t+1)=\tau(t) +1 | \tau(t) \leq  n\Big) = 1.
\end{equation}
as communication is never attempted for these values of $\tau(t)$.} For all $\ell \in \Zo$ we have
\begin{equation}
\begin{array}{lll}
     \Pr\Big(\tau(t+1)=n+ \ell +2 | \tau(t) = n+1+\ell\Big)& = &1-\eta   \\
 \Pr\Big(\tau(t+1)=1 | \tau(t) = n+1+\ell\Big)& = &\eta
\end{array}
\end{equation}
in view of SA\ref{ass:CSITpd} and the fact that $\eta$ denotes the packet success rate when $\tau(t) \geq n+1$. \textcolor{black}{Since $\tau(t)$ is always in one of these states, we have
\begin{equation}
  \Pr(\tau(t) \geq n+1) + \sum_{i=1}^n  \Pr(\tau(t)=i) = 1.
\end{equation}
Applying one of the basic rules of probability, we derive
\begin{equation}
 \eta^{-1}  \Pr(\tau(t) =1) + \sum_{i=1}^n  \Pr(\tau(t)=1) = 1.
\end{equation}
This allows us to evaluate
\begin{equation}
    \Pr(\tau(t)=1) = \frac{\eta}{n\eta+1},
\end{equation}
when the Markov chain is in steady state, which will also be the steady state probabilities for $\tau(t)=i$ for any $i \in \{1,\dots,n\}$.} Additionally, 
\begin{equation}
    \Pr(\tau(t) \geq n+1) = \frac{1-\eta}{n\eta+1}.
\end{equation}
Since the Markov chain is trivially irreducible (as $\tau(t)$ always cycles between states) and aperiodic for all $\eta \in (0,1)$, we have exponential convergence to the steady state distribution from \cite{rosenthal1995convergence}.
\end{proof}

\subsection{Proof of Proposition \ref{prop:etastab}}
\label{app:p1}
Given $\mu \in (a_S,\min\{1,a_U\})$, $n \in \Zo$ and $\eta \in [0,1]$, we first note that in view of SA\ref{ass:alpha1} and \eqref{eq:mainsys3b}, we have
\begin{equation}
    V(F(\chi,i+1)) \leq a_S a_U^i V(\chi) \label{eq:tta1a0}
\end{equation}
for all $\chi \in \mathbb{R}^{s_{\chi}}$ and $i \in \mathbb{Z}_{\geq 0}$.

Let $\chi_0 \in \mathbb{R}^{s_{\chi}}$ and consider $\chi(t)$ the solution to \eqref{eq:mainsys3b} initialized at $\chi_0$. Recall that due to the structure of \eqref{eq:TCTPC}, once a transmission is successful, the next transmission is attempted only after $n$ steps. Therefore, we define for all $t \in \Zo$
\begin{equation}
    \mathcal{T}_U(t) : = \Big\{i \in \{1,2,\dots,t-1\}\, |\, \tau(i+1)\geq 2+n\Big\},
\end{equation}
the set of all time instances where transmissions were attempted, but communication failed before $t$. This implies that for any $t \in \Zp$ and any $ i \in \mathcal{T}_U(t)$,
\begin{equation}
  V(\chi(i+1)) \leq a_U V(\chi(i)) \label{eq:setu}
\end{equation}
in view of SA\ref{ass:alpha1}.

On the other hand, we define the set of all time instances where transmission was successful before $t$ for all $t \in \Zo$ as
\begin{equation}
    \mathcal{T}_q(t) : = \Big\{i \in \{1,2,\dots,t-1\}\, |\, \tau(i+1)=1\Big\},
\end{equation}
because whenever a communication occurs at some time, we have $\tau(t+1)=1$ according to \eqref{eq:taud}. This allows us to use SA\ref{ass:alpha1} to write, for any $t\in \Zp$ and any $ i \in \mathcal{T}_q(t)  $,
\begin{equation}
   V(\chi(i+\ell)) \leq a_U^n a_S V(\chi(i)). \label{eq:sets}
\end{equation}
for all $\ell \in \{1,\dots,n+1\}$. Combining \eqref{eq:setu} and \eqref{eq:sets}, we can write
\begin{equation}
    V(\chi(t)) \leq a_S a_U^n\prod_{i=0}^{t-1} G(i) V(\chi_0) 
\end{equation}
where $G(i):=a_U$ if $i \in \mathcal{T}_U(t)$, $G(i):=a_S a_U^n$ if $i \in \mathcal{T}_q(t)$ and $G(i):=1$ otherwise. This can be done because we have $a_S \leq a_S a_U \leq \dots \leq a_S a_U^n$. Taking the logarithm on both sides, we have for any $t \in \Zp$,
\begin{equation}
    \log(V(\chi(t))) \leq \log( V(\chi(0)))+\sum_{i=1}^t \log(G(i))  .
\end{equation}

Note that under \eqref{eq:mainsys3b}, the clock state sequence can be seen as a Markov chain with steady state distribution as stated in Lemma \ref{lem:tau}. Recall that we initialize $\tau(1)=1$. This allows us to express $G(i)$ as a random variable, and its distribution can be calculated as follows,
\begin{equation}
    \begin{array}{lll}
    \Pr(G(i)=a_S a_U^n)     & \leq &\Pr(\tau(i+1)=1)  \\
     \Pr(G(i)=a_U)    &  \geq & \Pr(\tau(i+1)=n+1) 
    \end{array}
\end{equation}
for all $i \in \{0,\dots,t-1\}$ for any $t \in \mathbb{Z}_{>0}$.

The results of Lemma \ref{lem:tau} provides $\Pr(\tau(i+1)=1)$ and we have
\begin{equation}\begin{array}{rl}
  \Ex[\log(V(\chi(t)))]   \leq   &   \log( V(\chi_0)) \\
  & + t \Big( \Pr(\tau(t)>n+1) \log(a_U))\\
     & + \Pr(\tau(t)=n+1) \log(a_S a_U^n) \Big) \\
    \leq &  \log( V(\chi_0))+ t \beta(n,\eta)
\end{array}
\end{equation}

Taking the exponential on both sides, we get the convergence rate
\begin{equation}
    \Ex[V(\chi(t))] \leq \beta(n,\eta)^t V(\chi_0) \label{eq:munconv}
\end{equation}
Since $\beta(n,\eta)<\mu$, property \eqref{eq:munconv} automatically implies that
\begin{equation} \begin{array}{ll}
  \sum_{t=0}^\infty \Ex[\underline{\alpha}(|\chi(t)|) ]    & \leq       \sum_{t=0}^\infty \Ex[V(\chi(t))] \\
     &  \leq \frac{1}{1 - \mu} V(\chi_0) < \infty
     \end{array}
 \end{equation}
satisfying condition \eqref{eq:ststab} in Definition \ref{def:ss} as $\mu<1$ and concluding our proof. 

\subsection{Proof of Lemma \ref{prop:minpar}}
\label{app:p3}
\begin{proof}
Recall that we consider $a_U > a_S$ in SA\ref{ass:alpha1}. Due to the property of logarithms, if $\log(\beta(n,\eta))$ for any $n \in \Zo$, is monotonically decreasing in $\eta$, then so is $\beta(n,\eta)$. Taking the logarithm of \eqref{eq:muofp} on both sides, we obtain
\begin{equation}
    \log(\beta(n,\eta)) = \frac{\log(a_U)+ \log(a_S a_U^{n-1})\eta }{ 1+ n \eta}.
\end{equation}
Taking the derivative w.r.t. $\eta$, we have

\begin{equation} \begin{array}{c}
    \dfrac{\log(a_S a_U^{n-1})}{1+ n \eta} - \dfrac{n (\log(a_U)+ \log(a_S a_U^{n-1})\eta )}{(1+ n \eta)^2}  \\ 
      = \dfrac{ \log(a_S a_U^{n-1}) - n \log(a_U) }{(1+ n \eta)^2} = \dfrac{ \log(a_S)   -  \log(a_U) }{(1+ n \eta)^2} 
\end{array}
\end{equation}
which is negative as $\log(a_U)>\log(a_S)$. Therefore, $\beta(n,\eta)$ is monotonically decreasing in $\eta$.

Next, observe that we have $\beta(n,0)=a_U$. Since, we consider $\mu < a_U$, if $\beta(n,1)< \mu$, we have $\beta(n,1) \leq  \mu \leq \beta(n,0) $. Since $\beta(n,\cdot)$ is continuous by definition, there exists at least one $\eta$ such that $\beta(n,\eta)=\mu$. Finally, due to $\beta(n,\cdot)$ being monotonous, this $\eta$ is unique. Additionally if $\beta(n,1) > \mu$, then $\beta(n,\eta)>\mu$ for any $\eta \in [0,1]$. 

Finally, we look at $\eta_C(\bar{h},p)$ in \eqref{eq:etaCP} and notice that each term in the summation is increasing w.r.t. $p$ in view of item (i) SA\ref{ass:CSITpd}. Therefore, $\eta_C(\bar{h},\cdot)$ is an increasing function. A similar logic applies to $\eta_I(\bar{h},\kappa)$. 
\end{proof}

\subsection{Proof of Proposition \ref{th:optPnoCSI}}
\label{app:p4}
Since we know that $P(t)$ is a stochastic process under policy \eqref{eq:TCTPC}, we can rewrite the cost \eqref{eq:costperk} as
\begin{equation}
     J_{\mathrm{PT}}(p,n) = \mathbb{E} [P(t)] = p   \Pr(\tau(t) \geq n+1).
\end{equation}

Applying Lemma \ref{lem:tau}, we substitute for $\Pr(\tau(t) \geq n+1)$ which provides \eqref{eq:TTcost1}.

For $n=0$, we trivially have that the function $J_{\mathrm{PT}}(p,0)=p$, which is strictly increasing in $p$. For all other cases, we will have $\underline{p}_n>0$. In order to study the properties of $J_{\mathrm{PT}}(p,n)$ w.r.t $p$, we look at the properties of the inverse cost which is never zero for $p>0$ defined as
\begin{equation}
\xi_n(p) = \frac{1}{J_{\mathrm{PT}}(p,n)}  = \frac{1}{p} + n\frac{1-e(p)}{p}   
\end{equation}

Due to the stability requirement, we only look at $\xi_n(p)$ for all $p \in [\underline{p}_n,P_{\max}]$, $n \geq 1$. Note that due to item (i) of SA\ref{ass:CSITpd}, we have that $1-e(p)$ is a sigmoidal function of $p$. We can therefore apply Theorem 1 in \cite{rodriguez2003analytical}, to conclude that the term $\frac{1-e(p)}{p}$ is quasi-concave and takes the value $0$ at the limits when $p \to 0$ and $p \to \infty$. The term $\frac{1-e(p)}{p}$ therefore has a unique maximum at say $p^u$, is strictly increasing in the interval $(0,p^u)$ and is decreasing in the interval $(p^u, \infty)$. 

Now, we can consider the two cases.
\begin{enumerate}
\item There is no local extremum for $\xi_n(p)$ for $p >0$.
    \item There exists at least one $p^*$ which is a local extremum satisfying
\begin{equation}
    \frac{\partial \xi(p^*)}{\partial p} =  \frac{-n e'(p^*)}{p^*} - \frac{1+n(1-e(p^*)}{p^{*2}}  =0. \label{eq:powersadle}
\end{equation}
\end{enumerate}
    
In the first case, since $\xi_n(\cdot)$ is differentiable and has no local extremum, $\frac{\partial \xi(p)}{\partial p}$ is never $0$ for $p>0$. Note that the function $\xi_n$ is decreasing in the interval $(p^u,\infty)$ for any $n$, and so $p \mapsto \xi_n(p)$ must be decreasing for all $p >0$. Since $\xi_n(p)$ is differentiable and $\frac{\partial \xi(p^*)}{\partial p}$ is never $0$, $\xi_n(p)$ is always decreasing, which implies that $J_{\mathrm{PT}}(p,n)$ is always increasing.

For the second case, there exists at least one $p^*$ satisfying \eqref{eq:powersadle}. Then, we evaluate
\begin{equation}\begin{array}{ll}
 \dfrac{\partial^2 \xi(p)}{\partial p^2} =   & \dfrac{-n e''(p)}{p}    + \dfrac{2}{p^2} \left( \dfrac{1+n(1-e(p))}{p}   + n e'(p) \right)
\end{array}
\end{equation}

However, note that at a local extremum, the above expression will have the second term vanishing due to \eqref{eq:powersadle}, implying that
\begin{equation}\begin{array}{ll}
 \dfrac{\partial^2 \xi(p^*)}{\partial p^2} =   & \dfrac{-n e''(p^*)}{p^*}
\end{array}
\end{equation}
which is positive when $e$ is concave and negative when $e$ is convex. From item (ii) of SA\ref{ass:CSITpd}, we know that $(1-e)$ is initially convex and then concave. This means that $\xi$ has only local minima initially (when $1-e$ is convex), and then only local maxima. Since $\xi(p)$ is continuous and differentiable, this is only possible if the local minimum and maximum are unique. 

\bibliographystyle{unsrt} 
\bibliography{bib_global}\vskip 0pt plus -1fil

\end{document}